\documentclass[sigconf, nonacm]{acmart}

\PassOptionsToPackage{hyphens}{url}\usepackage{hyperref}
\usepackage[figuresright]{rotating} 
\usepackage{acro}

\usepackage{colortbl}
\usepackage{amsmath,amsfonts}
\usepackage{array}
\usepackage{algorithmic}
\usepackage{textcomp}
\usepackage{xcolor}
\usepackage{lscape}
\usepackage{pdflscape}
\usepackage{graphicx}
\usepackage{booktabs}
\usepackage{dirtytalk}
\usepackage{longtable}
\usepackage{wasysym}
\usepackage{multirow}
\usepackage{booktabs}
\usepackage{hyperref}
\usepackage{tikz}
\usepackage{tabularx}
\usepackage{url}
\usepackage{microtype}
\usepackage{diagbox}

\newcommand{\T}{\mathbb{T}}
\usetikzlibrary{turtle}
\usepackage[inline]{enumitem}

\NewDocumentCommand{\rot}{O{45} O{1em} m}{\makebox[#2][l]{\rotatebox{#1}{#3}}}

\newcommand{\oldstuff}[1]{}
\newcommand{\old}[1]{}

\newtheorem{definition}{Definition}[section]
\newtheorem{theorem}{Theorem} 

\def\BibTeX{{\rm B\kern-.05em{\sc i\kern-.025em b}\kern-.08em
    T\kern-.1667em\lower.7ex\hbox{E}\kern-.125emX}}

\makeatletter
\newcommand{\thickhline}{%
    \noalign {\ifnum 0=`}\fi \hrule height 1pt
    \futurelet \reserved@a \@xhline
}
\newcolumntype{"}{@{\hskip\tabcolsep\vrule width 1pt\hskip\tabcolsep}}
\makeatother

\makeatletter
\newcommand{\headerhline}{%
    \noalign {\ifnum 0=`}\fi \hrule height 1.5pt
    \futurelet \reserved@a \@xhline
}
\makeatother

\usepackage{tikz}
\newcommand*\emptycirc[1][1ex]{\tikz\draw (0,0) circle (#1);} 
\newcommand*\halfcirc[1][1ex]{%
  \begin{tikzpicture}
  \draw[fill] (0,0)-- (90:#1) arc (90:270:#1) -- cycle ;
  \draw (0,0) circle (#1);
  \end{tikzpicture}}
\newcommand*\fullcirc[1][1ex]{\tikz\fill (0,0) circle (#1);} 
\DeclareAcronym{defi}{
  short = DeFi,
  long  = decentralized finance
}

\DeclareAcronym{derec}{
  short = DeRec,
  long  = Decentralized Recovery
}

\DeclareAcronym{uto}{
  short = UTxO,
  long  = Unspent Transaction Outputs
}

\DeclareAcronym{mitm}{
  short = MiTM,
  long  = Man-in-The-Middle
}

\DeclareAcronym{dns}{
  short = DNS,
  long  = Domain Name System
}

\DeclareAcronym{dos}{
  short = DoS,
  long  = Denial-of-Service
}

\DeclareAcronym{nfc}{
  short = NFC,
  long  = Near Field Communication
}

\DeclareAcronym{ttl}{
  short = TTL,
  long  = time-to-live
}

\DeclareAcronym{adb}{
  short = ADB,
  long  = Android Debugging Bridge
}

\DeclareAcronym{dao}{
  short = DAO,
  long  = Decentralized Autonomous Organization
}

\DeclareAcronym{spa}{
  short = SPA,
  long  = Single Trace Power Analysis
}

\DeclareAcronym{hca}{
  short = HCA,
  long  = High-Correlation Analysis
}

\DeclareAcronym{rpc}{
  short = RPC,
  long  = Remote Procedure Call
}

\DeclareAcronym{tls}{
  short = TLS,
  long  = Transport Layer Security
}

\DeclareAcronym{ipc}{
  short = IPC,
  long  = Inter-Process Communication
}

\DeclareAcronym{see}{
  short = SEE,
  long  = Secure Execution Environment
}

\DeclareAcronym{tee}{
  short = TEE,
  long  = Trusted Execution Environment
}

\DeclareAcronym{nee}{
  short = NEE,
  long  = Not Secure Execution Environment
}

\DeclareAcronym{cnn}{
  short = CNN,
  long  = Convolutional Neural Network
}

\DeclareAcronym{ecc}{
  short = ECC,
  long  = Elliptic Curve Cryptography
}

\DeclareAcronym{ec}{
  short = EC,
  long  = Elliptic Curve
}

\DeclareAcronym{dsa}{
  short = DSA,
  long  = Digital Signature Algorithm
}

\DeclareAcronym{ecdsa}{
  short = ECDSA,
  long  = Elliptic Curve Digital Signature Algorithm
}

\DeclareAcronym{mpc}{
  short = MPC,
  long  = Multi-Party Computation
}
\usepackage{csquotes}
\usepackage{float}

\usepackage{pdflscape} 
\usepackage{xcolor}
\usepackage{hyperref}
\usepackage{longtable}
\usepackage{booktabs}
\usepackage{tabularx}
\usepackage{tabularx}
\usepackage{threeparttable}
\begin{document}

\title{Hedge Funds on a Swamp: Analyzing Patterns, Vulnerabilities, and Defense Measures in Blockchain Bridges}

 \author{Poupak Azad}
\affiliation{
  \institution{University of Manitoba}
  \country{Canada}
}
\email{azad@myumanitoba.ca}

\author{Jiahua Xu}
\affiliation{
  \institution{University College London}
  \country{United Kingdom}
}
\email{jiahua.xu@ucl.ac.uk}

\author{Yebo Feng}
\affiliation{
  \institution{Nanyang Technological University}
  \country{Singapore}
}
\email{yebo.feng@ntu.edu.sg}

\author{Preston Strowbridge}
\affiliation{
  \institution{University of Central Florida}
  \country{USA}
}
\email{pr438045@ucf.edu}

\author{Cuneyt Gurcan Akcora}
\affiliation{
  \institution{University of Central Florida}
  \country{USA}
}
\email{cuneyt.akcora@ucf.edu}

\begin{abstract}

  Blockchain bridges have become essential infrastructure for enabling interoperability across different blockchain networks, with more than \$24B monthly bridge transaction volume. However, their growing adoption has been accompanied by a disproportionate rise in security breaches, making them the single largest source of financial loss in Web3. For cross-chain ecosystems to be robust and sustainable, it is essential to understand and address these vulnerabilities. In this study, we present a comprehensive systematization of blockchain bridge design and security. We define three bridge security priors, formalize the architectural structure of 13 prominent bridges, and identify 23 attack vectors grounded in real-world blockchain exploits. Using this foundation, we evaluate 43 representative attack scenarios and introduce a layered threat model that captures security failures across source chain, off-chain, and destination chain components. 
  
  Our analysis at the static code and transaction network levels reveals recurring design flaws, particularly in access control, validator trust assumptions, and verification logic, and identifies key patterns in adversarial behavior based on transaction-level traces. To support future development, we propose a decision framework for bridge architecture design, along with defense mechanisms such as layered validation and circuit breakers. This work provides a data-driven foundation for evaluating bridge security and lays the groundwork for standardizing resilient cross-chain infrastructure.

\end{abstract}

\maketitle

\section{Introduction}
\label{sec:intro}

Blockchain bridges have become essential infrastructure in the blockchain ecosystem, enabling interoperability between otherwise isolated networks.  In its most basic asset-based bridge type, a bridge locks or burns assets on a source chain and mints or releases corresponding assets on a destination chain. This coordination preserves total supply across networks while facilitating asset mobility and composability. Bridges play a foundational role in enabling cross-chain decentralized applications, unifying fragmented liquidity pools, and allowing users to leverage features across heterogeneous platforms. For instance, monthly bridge transaction volume has exceeded \$24B \cite{defillama_bridges}, underscoring their systemic significance.

However, this growing importance has been matched by an alarming trend: bridges represent the single largest source of financial loss in Web3 security breaches.  As of mid 2025, 13 out of 39 bridges on \href{https://l2beat.com/bridges/summary}{l2beat.com} are already labelled as insecure.\footnote{%
The 13 bridges marked as vulnerable (indicated with a red cross and / or brown shield badge on the website) due to
unverified contracts and past hacks are: \textit{Aptos (LayerZero)}, \textit{LayerZero v2 OFTs}, \textit{Multichain}, \textit{Connext}, \textit{Omnichain (LayerZero)}, \textit{Hyperlane Nexus},
\textit{Socket}, \textit{Chainport}, \textit{Allbridge}, \textit{StarGate (LayerZero)}, \textit{Symbiosis}, \textit{Everclear}, and \textit{Orbit Bridge}.
} Of the top six bridges in total-value-locked as ranked by \cite{kiepuszewski2022rollups} in 2022, three (\textit{Multichain} \cite{multichain}, \textit{Ronin} \cite{ronin_bridge}, and \textit{Rainbow} \cite{rainbow_bridge}) have already been hacked for more than \$750M. A critical vulnerability in the sixth, \textit{Polygon's Plasma Bridge} \cite{polygon_plasma_bridge}, could have exposed \$850M; it was patched just in time after a white-hat disclosure that earned a \$2M bounty.  

Vulnerabilities range from smart contract bugs and improper verification logic to compromised multisig keys and failure-prone trusted validators. The result is a persistent threat model with catastrophic consequences: stealthy attacks, rapid fund drainage, and minimal recourse for victims. The threat of large-scale, unpredictable failure continues to undermine user trust and impede the adoption of cross-chain systems. With billions of USD locked in bridges, their current security risks make them resemble hedge funds operating atop a swamp; highly valuable, yet dangerously unstable and exposed to unpredictable attacks.

Despite a growing body of research and security reviews, there is still no unified framework to evaluate bridge vulnerabilities. Existing literature tends to focus on isolated case studies~\cite{augusto2024xchainwatcher, wu2024safeguarding}, postmortem hack reports~\cite{belenkov2025sok}, or abstract protocol taxonomies~\cite{li2025blockchain, belenkov2025sok}. This fragmentation limits our ability to draw general conclusions, compare implementations, or develop preventive standards.

In this work, we take a data interoperability perspective to present a comprehensive security and privacy systematization of blockchain bridges, grounded in both formal modeling and empirical analysis. We argue that bridge security is fundamentally a data transfer security problem, concerned with state consistency, transaction ordering, and trust-minimized replication across chains. We believe that it must be addressed with the same rigor as traditional interoperability standards. 

We approach the study of bridge security through three orthogonal lenses: 

First, we formalize the operational layers of bridge systems, spanning the source chain, off-chain intermediaries, and destination chain components. This layered model allows us to define three core security priors, such as cross-chain causality, and to reason about how different bridge designs instantiate or violate them. 

Second, we perform a large-scale static analysis of bridge smart contracts deployed on Ethereum \cite{ethereum}. Using program analysis and custom metrics, we map the usage of role-based access control, defensive programming patterns, error handling, and call structures across 13 major bridge protocols.  Next, we extract bridge-related transaction networks from on-chain data and apply graph machine learning to analyze usage and attacker behavior in real-world incidents. Our analysis spans 43 major attacks and reveals patterns in exploit phases, laundering strategies, and response delays.

Most importantly, we aim to answer four research questions to evaluate proactive goals: 
\begin{enumerate*}[label={\bf RQ\arabic*}]
\item (Prevention): To what extent are bridge attacks preventable? Which defense techniques, such as static code analysis, formal verification, or trusted execution environments, are effective? Are some bridge types inherently more secure?
\item (Detection): If prevention is not guaranteed, can attacks be detected in real time? What mechanisms (e.g., anomaly detection, on-chain monitoring) support timely
identification? 
\item (Mitigation): Once an attack is detected, what mechanisms exist to mitigate damage and minimize asset loss? How do mitigation strategies differ across bridge types?
\item (Outlook): What is the long-term outlook for bridge security? What foundational steps are needed to design formally verifiable or provably secure bridge protocols?
\end{enumerate*}

By unifying theoretical foundations with large-scale empirical analysis, this work offers the first comprehensive, data-driven systematization of blockchain bridge security. Our findings serve as a baseline for future research, standardization efforts, and the secure evolution of bridge protocols.

\textbf{Our contributions are as follows:}
\begin{itemize}
    \item We formalize a layered model of blockchain bridge architectures and define core security priors relevant to cross-chain behavior.
    \item We develop a vulnerability taxonomy and introduce a formal notion of bridge attack surfaces based on trust assumptions and implementation details.
    \item We perform the first large-scale static analysis of deployed bridge contracts, quantifying security patterns across access controls, call structures, and guard mechanisms.
    \item We extract and analyze transaction-level behavior from bridge exploit incidents, identifying behavioral patterns across phases of use, compromise, and fund laundering.
    \item We provide security benchmarks and design recommendations to guide future bridge development, formal analysis, and regulatory evaluation.
\end{itemize}

\section{Related Work}
Recent literature reflects an active effort to reconcile decentralization, trust minimization, and performance in blockchain bridges. We summarize case studies and surveys here and refer the reader to Appendix~\ref{sec:extendedRelated} in the supplementary material (available at our \href{https://github.com/FDataLab/BlockchainBridgeAnalysis/blob/main/README.md}{repository URL}) for additional details on interoperability and benchmark studies.

Transaction analysis for bridges can identify operational patterns, including usage behaviors, anomalous activity, and indicators of compromise. Huang et al.~\cite{huang2024seamlessly} present an in-depth empirical analysis of Stargate, a prominent Layer-0 bridge with one of the highest total value locked (TVL) figures in the ecosystem. Using on-chain data from six EVM-compatible blockchains (Ethereum, Polygon \cite{polygon}, BSC~\cite{bsc2020whitepaper}, Avalanche \cite{avalanche}, Arbitrum \cite{arbitrum}, Optimism \cite{optimism}), they examine Stargate’s transaction volume, user adoption, and operational patterns. This case study provides a data-driven benchmark for analyzing cross-chain liquidity transfers in practice. However, the study does not identify bridge priors nor carry out a static code analysis.

Subramanian et al.~\cite{subramanian2024benchmarking} benchmark the performance of blockchain bridge aggregators: services that route transfers across multiple bridge protocols for optimal speed or cost. They develop a framework to test popular aggregators (e.g., LI.FI, Socket, deBridge) by executing hundreds of cross-chain swaps and recording metrics such as fees, slippage, and latency. Their findings quantify differences in cost-efficiency and reliability, providing performance benchmarks for user-centric interoperability services. Our work differs in its broader focus and the development of a taxonomy to offer a holistic view of bridges.

Augusto et al.~\cite{augusto2025xchaindatagen} introduce \textit{XChainDataGen}, a framework for generating large-scale bridge transaction datasets. Using this tool, they collect approximately 35~GB of data from five major bridges deployed on 11 blockchains during the second half of 2024, extracting over 11.2 million bridge transactions that moved over \$28B in token value. They compare protocols in terms of security, cost, and performance, contrasting, for example, full source-chain finality with \textquote{soft} finality, alongside fee models and emerging paradigms such as cross-chain intents. However, the study focuses on financial aspects and does not study bridge vulnerabilities that hinder bridge adoption in finance.

Recent surveys and taxonomies have focused on identifying systemic vulnerabilities in bridge designs. Li et al.~\cite{li2025blockchain} provide a comprehensive review of blockchain bridges, classifying architectures (e.g., lock-mint vs. notary schemes) and documenting common vulnerabilities such as smart contract flaws, centralization risks, liquidity issues, and oracle manipulation. While the study presents a broad taxonomy and useful design insights, it does not analyze transaction data or investigate real-world exploits. In a recent work, Augusto et al.~\cite{augusto2024sok} analyze bridge exploits; however, they do not conduct a direct analysis of what happens to stolen funds after bridge hacks, nor carry out static code analysis themselves. Instead, they rely heavily on systematic literature review, audit reports, bug bounty disclosures, and gray literature to collect information about vulnerabilities and attacks.  Our analysis covers both the code and transaction analysis.

In contrast, Belenkov et al.~\cite{belenkov2025sok} present a Systematization of Knowledge focused on major bridge hacks, including the \$600M Axie \textit{Ronin} exploit. They categorize failure modes such as compromised keys, multisig errors, and validator logic flaws, and propose best practices for prevention. However, the analysis is limited to static perspectives and excludes transactional behavior. Our work complements and extends these efforts by combining both static and transaction analysis perspectives and grounds them in a unified formal model.

Most security and forensic analyses of bridge hacks have been retrospective machine learning studies of transaction networks, typically conducted on a per-case basis and without developing generalizable vulnerability taxonomies or formalizing bridge-specific assumptions. For example, Augusto et al.~\cite{augusto2024xchainwatcher} propose \textit{XChainWatcher}, a Datalog-based monitoring system that detects anomalous token flows across bridges in real time. While it successfully reconstructs the \textit{Ronin} and \textit{Nomad} \cite{nomad} attacks, its scope is limited to two case studies and does not address broader questions about interoperability security. Similarly, Lin et al.~\cite{lin2025track} develop \textit{ABCTracer}, an automated framework for linking bridge transaction legs. Combining event log mining with machine learning inference, \textit{ABCTracer} achieves 91.75\% F1-score in identifying transaction pairs across 12 DeFi bridges. While valuable for forensics, the work does not define a vulnerability taxonomy nor investigate design assumptions underpinning bridge protocols.

Finally, Wu et al.~\cite{wu2024safeguarding} focus on identifying and classifying bridge-specific attacks. Using a dataset of 49 bridge exploits totaling nearly \$4.3B in losses, they propose \textit{BridgeGuard}, a graph-based detection system that models bridge transactions and flags anomalous behavior. Evaluated on 203 known attacks and 40,000 benign transactions, \textit{BridgeGuard} outperforms prior methods in recall and detection precision. However, their work does not incorporate static analysis of contract logic or propose a general framework for understanding bridge security.

While these studies contribute valuable empirical and forensic insights, they tend to focus narrowly on individual case studies, static audits, or retrospective detection. In contrast, our work presents the first unified and multi-dimensional study that combines a formal model of bridge semantics, a systematic static analysis of deployed bridge contracts, and a large-scale transaction analysis across multiple chains.  

\section{Bridge Modeling and Formalization} \label{sec:background}
A blockchain $b$ is an immutable ledger where transactions are appended chronologically. It comprises a set of peer-to-peer network nodes $N_p$, a set of address nodes $N_a$, a chronological history of transactions $\mathbb{H}$ between nodes in $N_a$, and consensus rules $R$ that govern transaction creation, where $N_p \cap N_a \neq \emptyset$ and $b = (N_p, N_a, \mathbb{H}, R)$. 

\paragraph{Asset Types} A clear distinction exists between coins, tokens, and wrapped assets. Coins are native digital currencies intrinsic to their own blockchains, such as Bitcoin  (BTC) on the Bitcoin network \cite{bitcoin}, and are essential for paying transaction fees and participating in consensus mechanisms. A smart contract-based token $\theta$ represents a type of blockchain currency, with its state defined by the chronological history of transactions $\mathbb{H}_{\theta}$ in which $\theta$ has been involved. All blockchains issue native coins (e.g., Ether on Ethereum), but only some blockchains allow users to issue smart contract-based tokens (e.g., Storj on Ethereum). 

Wrapped tokens are a subset of tokens designed to represent tokens from one blockchain on another, facilitating interoperability across disparate blockchain networks. For instance, Wrapped Bitcoin (WBTC) is issued on Ethereum, enabling BTC to be utilized within Ethereum's decentralized finance ecosystem. Despite enhancing bridge functionality, wrapped tokens inherit the limitations of standard tokens: they cannot be used to pay transaction fees on the host blockchain (i.e., Ethereum for WBTC).  

A transaction $tx \in \mathbb{H}$ on $b$ transfers a value $v$ of a token $\theta$ from an address $a_1 \in N_a$ to an address $a_2 \in N_a$, and is associated with a timestamp $t_{tx}$ indicating when the transaction occurred: 
$tx = (\theta, v, a_1, a_2, t_{tx})$. If the transaction is understood, we will simplify $t_{tx}$ to $t$. Each transaction is assumed to be valid under the consensus rules $R$ of $b$.

\paragraph{Blockchain Layers} Blockchain systems are structured into hierarchical layers to manage scalability, functionality, and specialization.

\begin{itemize}[leftmargin=*]
    \item \textbf{Layer 1 (L1):} Base layer of a blockchain network, encompassing the core protocol responsible for transaction validation, consensus mechanisms, and data storage. L1 blockchains operate independently and are the foundation upon which other layers are built. Examples include Bitcoin, Ethereum, and Solana \cite{solana}.
    
    \item \textbf{Layer 2 (L2):} Built atop L1 blockchains, L2 protocols aim to enhance scalability and transaction throughput without altering the base protocol. They achieve this by processing transactions off-chain or in parallel, subsequently settling them on the L1. Notable L2 implementations include Optimism and Arbitrum on Ethereum.
    
    \item \textbf{Layer 3 (L3):} L3 represents an emerging concept focusing on application-specific functionalities. These are protocols or networks constructed on L2 solutions, offering tailored environments for decentralized applications (dApps). L3s aim to provide enhanced scalability, interoperability, and customization, facilitating complex applications like decentralized finance platforms and gaming ecosystems. Examples of L3 projects include Orbs Network \cite{orbs_network} and XAI Games \cite{xai_games} on Arbitrum.
\end{itemize}

\subsection{Blockchain Bridges}
A blockchain bridge facilitates the transfer of tokens across distinct blockchain domains. While many bridges connect L1 blockchains, others operate across layers, such as the \textit{Arbitrum Canonical Bridge} between Ethereum (L1) and Arbitrum (L2). Bridges that involve L3 layers are an emerging area with early examples such as \textit{Arbitrum Orbit}, \textit{zkSync Hyperchains} \cite{zksync}, and \textit{Orbs Network}.  Despite their growing importance, blockchain bridges lack a consistent formalization in the literature, which we aim to define as follows.

We consider two blockchain domains (i.e., chains or protocols), $b_1\in \{L1, L2, L3\}$ and $b_2\in \{L1, L2, L3\}$, which may differ in node sets, transaction histories, and consensus rules. We define a bridge between the domains $b_1$ and $b_2$ through an implementation mechanism $\mathcal{I}$ that governs the operational logic of the bridge:

\begin{equation}
\mathbb{B}_{1 \leftrightarrow 2} = (\{b_1, b_2\}, \mathcal{I})
\end{equation}

We categorize bridges along two orthogonal axes: their \textit{operational trust model} and their \textit{functional type}. The trust model refers to how source domain activity is verified and includes \textit{trusted}, \textit{trust-minimized}, and \textit{trustless} designs, formalized shortly in Section~\ref{sec:formaliseattacksurface}. The functional type captures how value is transferred and includes:

\begin{itemize}[leftmargin=*]
    \item \textbf{Asset Bridges:} In \textit{asset bridges} (also known as burn-and-mint models), the bridge locks or burns an asset $\theta_1$ on the source domain $b_1$ and creates a representative asset $\theta_2$ on the destination domain $b_2$. These assets are not inherently equivalent; their linkage is established solely through the semantics of the bridge protocol. They include:
    \begin{itemize}[leftmargin=*]
        \item \textit{Externally-Verified Asset Bridges}, which depend on multisigs, notaries, or sidechains to verify lock events (e.g., \textit{Wormhole} \cite{wormhole}, \textit{Ronin}, \textit{Multichain}).
        \item \textit{Rollup-Native Asset Bridges}, which leverage L1 consensus to verify L2 state transitions using fraud or validity proofs (e.g., Arbitrum, Optimism, Loopring \cite{loopring}). These are considered the most trustless.
    \end{itemize}
    The settlement process on asset-based bridges may be slow and require specialized verification and challenge code. For example, optimistic rollups like Arbitrum or Optimism impose a 7-day challenge period on withdrawals.
    \item \textbf{Liquidity Networks:} In contrast to asset-based bridges, \textit{liquidity network-based bridges}, such as Connext\cite{Connext2025} or Across\cite{Across2025}, avoid minting and instead fulfill user requests through liquidity providers (LPs) who maintain reserves on both domains. LPs are economically incentivized via transfer fees and, in some protocols, additional yield or reward mechanisms. This design enables faster settlement and circumvents the latency of on-chain verification and challenge periods. 
    
    \item \textbf{Hybrid Bridges:} These support both functional models, often depending on the specific chain pair or asset type. For instance, \textit{Multichain} combines canonical minting with liquidity provisioning.
\end{itemize}

A bridge transaction $tx$ moves a value $v$ of token $\theta_1$ from an address $a_1 \in N_a$ on $b_1$ to an address $a_2 \in N_a$ on $b_2$:
\begin{equation}
tx = (\theta_1, v, a_1, a_2, t_{tx})
\end{equation}

The transformation of a bridge transaction state $\chi$ under the bridge’s operation is defined as:%
\begin{equation}
\chi_0 \xrightarrow{\mathbb{B}_{1 \leftrightarrow 2}} \chi_t
\end{equation}%
where $\chi_0$ and $\chi_t$ represent the initial and final states of the transaction across chains.

The implementation $\mathcal{I}$ of the bridge includes mechanisms on the source chain, off-chain components, and the destination chain. Additionally, the bridge relies on a node trust set $\mathbb{T}$, comprising entities responsible for validating and securing bridge transactions. The composition of $\mathbb{T}$ is determined by the specific implementation $\mathcal{I}$ and the security assumptions of $\{b_1, b_2\}$.

\subsection{Bridge Mechanism}
Consider that Alice wants to move some assets from the source chain $ b_1$ to the destination chain $ b_2$. For simplicity in exposition, we will assume that at a given time, the bridge is used by one user only and for bridging one token only. The token Alice will be moving is $\theta_1$, and the representation of the token on $ b_2$ is $\theta_2$. We track the movement of value across the bridge $\mathbb{B}_{1 \leftrightarrow 2}$ as a function of time and fee structure. The state of the bridge transaction $\chi$ is described by four accounts: Alice's addresses $a_{1}$ on $b_1$ and $a_2$ on $b_{2}$ along with the bridge addresses $c_{1}$ on $b_1$ and $c_{2}$ on $b_2$. The state of the bridge transaction is represented as 
$\chi = (a_{1}, c_{1}, a_{2}, c_{2})$.

We track the mechanisms in three stages: source chain mechanism, off-chain mechanism, and destination chain mechanism.

\subsubsection{Source Chain Mechanism} 
Alice initiates the transfer on $ b_1$. Initially, at $t = 0$, Alice holds $v_1$ units of the token $\theta_1$ at her address $a_{1}$, with a total value $v_1 \cdot price(\theta_1,t)$, where $v_1$ is the amount of the token, and $price(\theta_1,t) $ represents the price of the token $\theta_1$ in a fiat currency such as USD at time $t$. We will use $a_x \mapsto v$ to show that the balance of address $a_x$ is $v$. Hence, the initial state of balances is $\chi_0 = (a_1 \mapsto  v_1, c_1 \mapsto  0, a_2 \mapsto  0, c_{2} \mapsto 0)$.

If Alice wants to transfer an amount $v_x \leq v_1$ to her address $ a_{2}$ on the other blockchain, she initiates the transaction $tx_{b_1} = (\theta_1, v_x, a_{1}, c_{1}, t_{tx_{b_1}})$ on $ b_1$.

She also pays a bridge fee $F_\text{forward}$, which consists of the transaction processing fee $f_1$ on $ b_1$ and $f_2$ on $b_2$, as well as the bridge operation fee ($f^*$). Hence the total fee $F_\text{forward} = f_1 + f_2 + f^*$. For simplicity, we assume that the bridge operation fee is charged on the initiating blockchain $ b_1$. The sent amount $v_x$ and fees are deducted from the initial balance of $a_1 \mapsto v_1$, hence $a_1 \mapsto  (v_1 - v_x - f_1 - f^*)$. This leaves only $f_2$ to be paid on $b_2$.

After $ b_1$'s block confirmation duration $d_{b_1}$, the 
bridge smart contract on $ b_1$ acknowledges $tx_{b_1}$ and locks ($c_1\mapsto v_x$) or burns ($c_1\mapsto 0$) the received assets depending on the implementation $\mathcal{I}$. We follow the locked asset scenario, and the state is updated as $\chi_{t_{tx_{b_1}}} = (a_1 \mapsto  (v_1 - v_x - f_1 - f^*), c_1 \mapsto  v_x, a_2 \mapsto  0, c_2 \mapsto  0)$. The blockchain state changes as $ b_1 = (N_p, N_a, \mathbb{H} \cup \{tx_{b_1}\}, R)$.

\subsubsection{Off-chain Mechanism} 
An off-chain communicator observes transactions to bridge address $c_{1}$, specifically transactions of the form: 
$tx_{B_1} = (\theta, v, a, c, t)$. 
Upon detection, it instructs $c_{2}$ to mint $v_x-f_2$ worth of $\theta_2$. The off-chain mechanism takes time $d_\text{off}$ to notice the first transaction and send a signal to the $ b_2$. Thus, at 
$t_\text{off} = t_{tx_{b_1}} + d_\text{off}$, the off-chain mechanism signals the start of the token transfer process on the second blockchain $ b_2$. Specifically, the mechanism signals the smart contract on $ b_2$ accordingly. There are multiple mechanisms to create this signal, and we refer the reader to \cite{belenkov2025sok} for a comprehensive list. In a simple oracle-based solution, a network of oracles monitors transactions on blockchain $ b_1$. These relayers are either permissioned (managed by a specific entity) or decentralized (e.g., a set of staked validators). When a bridge transaction $tx_{b_1}$ is observed, relayers submit proofs or messages (through transactions that push data to the blockchain) to a contract on $ b_2$ to trigger the minting or unlocking of assets. The node trust set $\mathbb{T}$ validates these processes based on the implementation $\mathcal{I}$.

\subsubsection{Destination Chain Mechanism} 
The transfer completes when $v_x-f_2$ worth of $\theta_2$ is minted on $ b_2$ and sent to Alice’s address $a_{2}$: 
$tx_{b_2} = (\theta_2, v_2=0, a_{2}\mapsto v_x-f_2, c_{2}\mapsto 0, t_\text{off})$. After $b_2$’s block confirmation duration $D_{b_2}$, the complete bridge process ends at $t= t_{tx_{b_1}} + d_\text{off} + d_{b_2}$. 
The final state is: 
$\chi_t = (a_1 \mapsto  (v - v_x-f_1-f^*), c_1 \mapsto  0, a_2 \mapsto  v_x -f_2, c_2 \mapsto  0)$. 

The destination domain is $ b_2 = (N_p, N_a, \mathbb{H} \cup \{t_{b_2}\}, R)$. 
The smart contract on $ b_2$ handles the transaction and updates the blockchain state accordingly.

\subsubsection{Reverse Process} The reverse transfer follows the same mechanism as the forward transfer but in the opposite direction. Alice initiates a transaction on $ b_2$ to send $v_x$ of $\theta_2$ back to $ b_1$. The bridge smart contract on $ b_2$ locks ($c_2\mapsto v_x$) or burns ($c_2\mapsto 0$) token $\theta_2$, and after confirmation, an off-chain mechanism signals $ b_1$ to release $(v_x - F_{\text{reverse}})$ of $\theta_1$ to Alice’s address $a_1$. The final state transition mirrors the forward process, with updated fees and timestamps.


\subsection{Formalization of Bridge Implementations}
\label{sec:formaliseimplement}

Blockchain bridges can be categorized based on their trust assumptions, yielding three classes: trustless, trusted, and trust-minimized. A trustless bridge has no external trusted entities, with $\T_{\textit{trustless}} = \{\}$, relying exclusively on the security guarantees of the underlying blockchains. A trusted bridge introduces external dependencies, requiring $\T_{\textit{trusted}} \neq \{\}$. A trust-minimized bridge is a subset of trusted bridges in which all trusted entities $\T_{\textit{trustmin}}$ consist of deterministic, publicly auditable algorithms (for example, on-chain smart contracts, oracles, or relays). While trust-minimized bridges rely on external components, their correctness can be independently verified by users.

The bridge trust set is expressed as
$\T = \T_{\textit{src}} \cup \T_{\textit{off}} \cup \T_{\textit{dest}},$
where $\T_{\textit{src}}$ captures trust assumptions on the source blockchain, $\T_{\textit{off}}$ captures trust in off-chain mechanisms (such as relayers or validators), and $\T_{\textit{dest}}$ captures trust in the destination blockchain. Two key metrics, $Size(\T)$ and $ cost(\T)$, characterize the scale of the trusted set and its overall fee impact, respectively. These metrics influence the security, decentralization, and efficiency of a bridge.

\subsubsection{Source Chain Implementation for Trust Set $\T_{\textit{src}}$}
The source chain component locks or burns assets before triggering $\T_{\textit{off}}$. It can be realized via smart contracts (so $\T_{\textit{src}}=\{SCs\}$), validator-controlled mechanisms ($\T_{\textit{src}}=\{\}$), or hybrid approaches that combine both. Smart contract mechanisms require trusting the correctness of the contract logic, whereas validator-based methods rely on the blockchain’s consensus, eliminating additional trust. Hybrid methods usually trade off fees or time for scalability and security. As shown in Table~\ref{tab:classification_table}, these components appear on both source and destination chains, but since their role is more critical on the destination side, where bridged tokens are released, we omit their detailed explanation here and defer it to the destination chain discussion to avoid repetition.

\subsubsection{Off-chain Implementation for Trust Set $\T_{\textit{off}}$}
Bridges typically use either a set of notaries ($N$) or a set of light clients ($L$) for off-chain processing.

\paragraph{Notaries}
Notaries introduce additional trust in external parties ($\T_{\textit{off}}\neq\{\}$) that monitor transactions and ensure bridge causality, with security depending on the number of notaries $Size(N)$ and their operational costs $cost(N)$. As $Size(N)$ and $ cost(N)$ grow, security improves, but so do fees and transaction delays. The expected loss costs $\mathbb{E}(C)$, bridge fees $F$, and transaction time delay $D$ follow
$\mathbb{E}(C)\propto 1/Size(N)$ and $F,D\propto Size(N),Cost(N)$.

\begin{figure}[ht]
\centering
\includegraphics[width=1\linewidth]{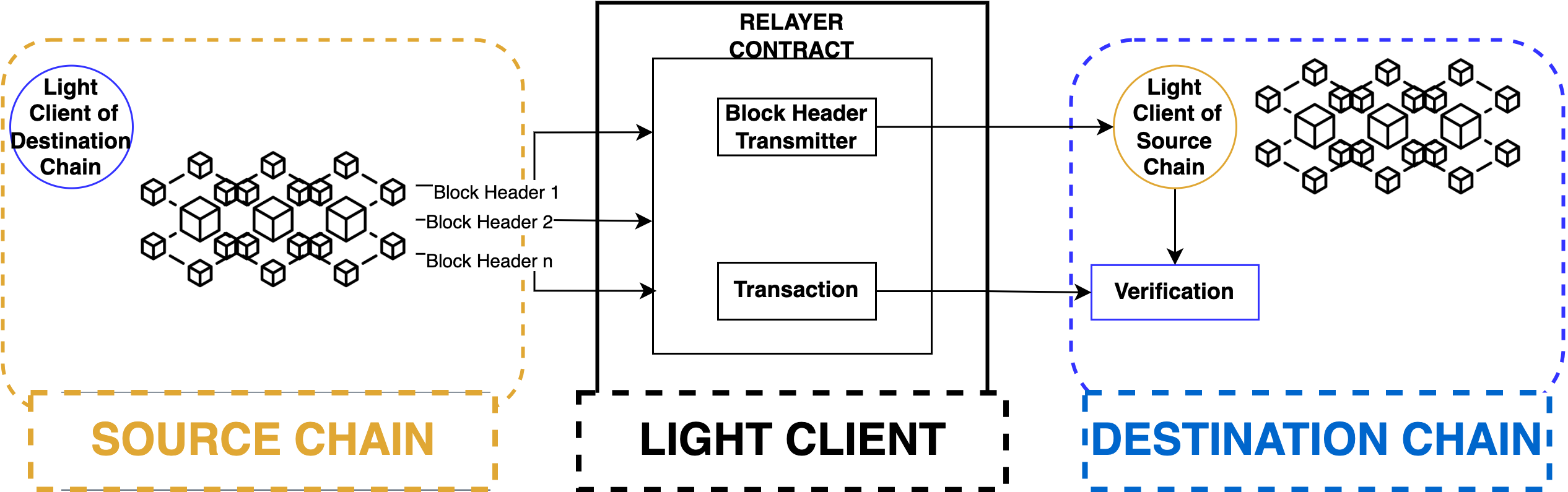}
\caption{Light Client Implementation}
\label{fig:lc_mech}
\end{figure}

\paragraph{Light Clients} Light clients rely on Merkle proofs ($M$) and perform verification on-chain via light client contracts ($L$). These systems do not depend on external entities, so the off-chain trust set is $\T_{\textit{off}} = \{L, M\}$. These bridges are trust-minimized because no external entities are needed beyond the blockchain itself. Verification time $t_{\textit{proof}}$ affects expected loss costs $\mathbb{E}(C)$ and bridge cost $F,D$ according to
$\mathbb{E}(C)\propto 1/t_{\textit{proof}},\quad F,D\propto t_{\textit{proof}}.$
Although frequent proof generation boosts security, it also increases fees. Light clients must run on each chain, reducing overall scalability. Figure~\ref{fig:lc_mech} illustrates a standard light client implementation used in many trustless bridge designs.

Some bridges blend the two approaches. A security-optimized hybrid requires that either the light client or the notary set remain honest, so
$\T_\text{off}=\{L,M\}\cap\{N\}.$
Another hybrid uses light clients where available and defaults to notaries otherwise:
$\T_\text{off}=\{L,M\}\vee\{N\}.$

\paragraph{Sidechains} A third option is using sidechains, which function as independent blockchains that can facilitate bridge transfers. Sidechains add a secondary blockchain with its own consensus $R^\prime$. A bridge passing through a sidechain depends on
$\T_\text{off}=(\{L,M\}\vee\{N\}\vee(\{L,M\}\cap\{N\}))\cup R^\prime.$

In other words, instead of relying exclusively on a notary network or on-chain light clients to handle the bridging, a sidechain can be set up with built-in asset-transfer features. Bridge transactions then pass through that sidechain, inheriting its consensus and security properties.

While the idea of a sidechain can still incorporate notaries or light clients under the hood, the main distinction is that a sidechain is an entire blockchain in its own right, rather than a narrowly defined tool like a light client or a notary service. This extra layer can increase scalability (since the sidechain can process many transactions internally) but also introduces new trust assumptions. Specifically, the correctness and security of the sidechain’s consensus mechanism. Figure~\ref{fig:sc_mech} illustrates a representative sidechain-based bridge design.

\begin{figure}[ht]
\centering
\includegraphics[width=1\linewidth]{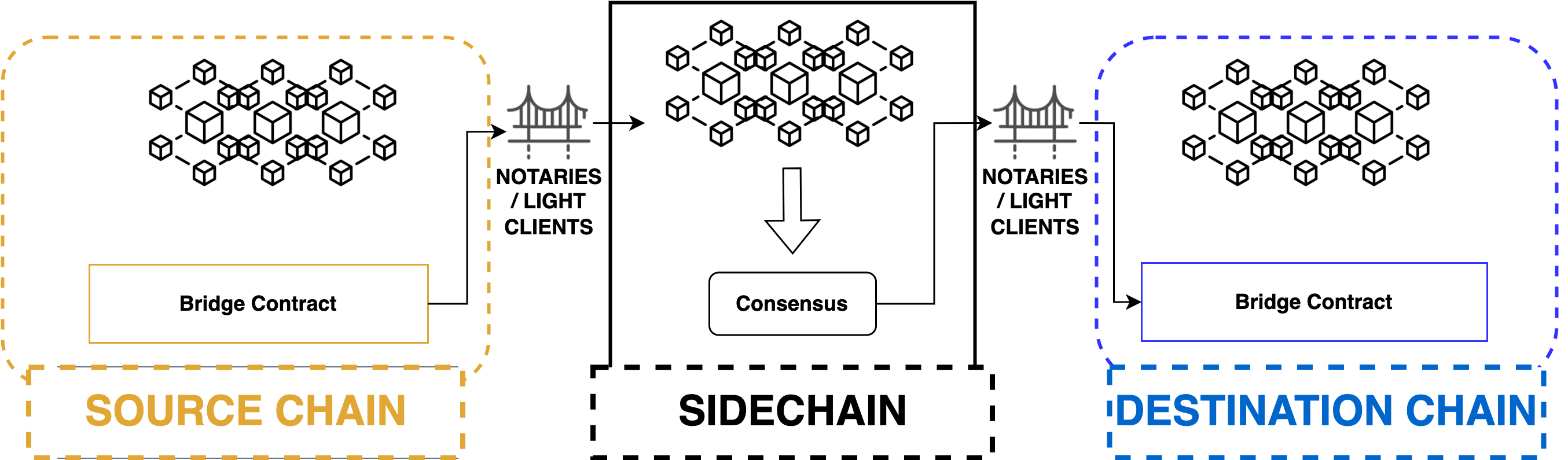}
\caption{Sidechain Implementation}
\label{fig:sc_mech}
\end{figure}

\subsubsection{Destination Chain Implementation $\T_{\textit{dest}} $}
\label{subsec:destimplement}

The destination side enforces token minting, final validation, and recipient assignment. Possible methods include smart contracts ($\T_{\textit{dest}} = \{SCs\} $), validator-control ($\T_{\textit{dest}} = \{\} $), custodian-based ($\T_{\textit{dest}} = \{C\} $), or a hybrid that integrates these techniques. The choice depends on constraints related to cost, security, and overall efficiency.

\paragraph{Smart Contracts}
A smart contract can be deployed on the destination chain to handle final token transfers. After receiving a valid signal and data from $\T_{\textit{off}}$, the contract locks, mints, or releases tokens to the designated address. Because users must trust the correctness of the contract code, we have $\T_{\textit{dest}} = \{SC\}$.

\paragraph{Validator Control}
If the destination blockchain’s existing consensus participants validate the bridging process, no additional trusted entities are introduced. The bridge may still have predefined logic to identify the destination address, but final checks are performed by decentralized validators. Thus $\T_{\textit{dest}} = \{\}$.

\paragraph{Hybrid}
A hybrid model combines smart contracts with validator-based consensus to balance security, cost, and time. The exact composition of $\T_{\textit{dest}} $ depends on how responsibilities are allocated between on-chain code and validator consensus.

\begin{table}
\caption{Cross-chain bridge classification for the highest volume bridges. \fullcirc \hspace{0.5mm}: trusted, \halfcirc \hspace{0.5mm}: trust minimized, \emptycirc \hspace{0.5mm}: trustless.}
\label{tab:classification_table}
\scriptsize
{%
\begin{tabular}{@{}llllc@{}}
\toprule
 &
  \multicolumn{1}{c}{\textbf{Bridge}} &
  \multicolumn{1}{c}{\textbf{\begin{tabular}[c]{@{}c@{}}Source \\
  Chain\end{tabular}}} &
  \multicolumn{1}{c}{\textbf{\begin{tabular}[c]{@{}c@{}}Destination \\
  Chain\end{tabular}}} &
  \multicolumn{1}{c}{\textbf{Trust}} \\
\midrule
{\multirow{6}{*}{\textbf{\begin{tabular}[c]{@{}l@{}}Notaries \ \end{tabular}}}} &
  {Nomad Bridge} &
  {Hybrid} &
  {Smart Contracts} &
  {\fullcirc} \\
{} &
  {Allbridge Classic} &
  {Smart Contracts} &
  {Smart Contracts} &
  {\fullcirc} \\
{} &
  {deBridge} &
  {Smart Contracts} &
  {Smart Contracts} &
  {\fullcirc} \\
{} &
  {Multichain} &
  {Smart Contracts} &
  {Smart Contracts} &
  {\fullcirc} \\
{} &
  {Wormhole Bridge} &
  {Smart Contracts} &
  {Smart Contracts} &
  {\fullcirc} \\
{} &
  {Avalanche Bridge} &
  {Smart Contracts} &
  {Validator Control} &
  {\fullcirc} \\
{} &
  {Ronin Bridge} &
  {Smart Contracts} &
  {Smart Contracts} &
  {\fullcirc} \\
{} &
  {Wanchain Bridge} &
  {Smart Contracts} &
  {Smart Contracts} &
  {\fullcirc} \\
\midrule
{\multirow{4}{*}{\textbf{\begin{tabular}[c]{@{}l@{}}Light \\
Client\end{tabular}}}} &
  {BTC Relay} &
  {Validator Control} &
  {Smart Contracts} &
  {\halfcirc} \\
{} &
  {zkBridge} &
  {Validator Control} &
  {Validator Control} &
  {\halfcirc} \\
{} &
  {Rainbow Bridge} &
  {Validator Control} &
  {Validator Control} &
  {\halfcirc} \\
{} &
  {PeaceRelay} &
  {Smart Contracts} &
  {Smart Contracts} &
  {\halfcirc} \\
\midrule
{\multirow{3}{*}{\textbf{\begin{tabular}[c]{@{}l@{}}Hybrid \ \end{tabular}}}} &
  {Connext} &
  {Smart Contracts} &
  {Smart Contracts} &
  {\fullcirc} \\
{} &
  {Optics Bridge} &
  {Smart Contracts} &
  {Smart Contracts} &
  {\fullcirc} \\
\midrule
{\multirow{15}{*}{\textbf{\begin{tabular}[c]{@{}l@{}}Sidechain \ \end{tabular}}}} &
  {Cosmos IBC} &
  {Validator Control} &
  {Validator Control} &
  {\emptycirc} \\
{} &
  {Gravity Bridge} &
  {Smart Contracts} &
  {Smart Contracts} &
  {\halfcirc} \\
{} &
  {zkRelay} &
  {Smart Contracts} &
  {Smart Contracts} &
  {\halfcirc} \\
{} &
  {Cactus} &
  {Smart Contracts} &
  {Smart Contracts} &
  {\halfcirc} \\
{} &
  {Celer cBridge} &
  {Smart Contracts} &
  {Smart Contracts} &
  {\halfcirc} \\
{} &
  {Orbit Bridge} &
  {Smart Contracts} &
  {Smart Contracts} &
  {\halfcirc} \\
{} &
  {Axelar} &
  {Smart Contracts} &
  {Smart Contracts} &
  {\halfcirc} \\
{} &
  {Chainswap} &
  {Smart Contracts} &
  {Smart Contracts} &
  {\halfcirc} \\
{} &
  {PolyBridge} &
  {Smart Contracts} &
  {Smart Contracts} &
  {\halfcirc} \\
{} &
  {pNetwork} &
  {Smart Contracts} &
  {Smart Contracts} &
  {\halfcirc} \\
{} &
  {Meter Passport} &
  {Smart Contracts} &
  {Smart Contracts} &
  {\halfcirc} \\
{} &
  {QANX Bridge} &
  {Smart Contracts} &
  {Smart Contracts} &
  {\halfcirc} \\
{} &
  {Binance Bridge} &
  {Smart Contracts} &
  {Smart Contracts} &
  {\halfcirc} \\
{} &
  {Horizon Bridge} &
  {Smart Contracts} &
  {Smart Contracts} &
  {\halfcirc} \\
{} &
  {Plasma Bridge} &
  {Smart Contracts} &
  {Smart Contracts} &
  {\halfcirc} \\
\bottomrule
\end{tabular}%
}
\end{table}

\subsection{Bridge Security Priors} 
\label{bridge_security}

We formalize three key security properties that bridges must satisfy to ensure token parity, transaction causality, and value preservation across blockchains. A fourth foundational requirement is the liveness of the bridge, but we assume that the underlying blockchains are live and capable of processing transactions in a timely manner, allowing us to focus on security properties specific to the bridge itself.

\paragraph{Bridge Peg} Blockchain domains $b_1$ and $b_2 $ often have disparate implementations of tokens. For users to reliably use token $\theta_1$ on $b_1$ as a representation of token $\theta_2$ on $b_2$, the bridge must establish that the asset amounts remain equivalent, i.e., $v_1 \equiv v_2$. This holds only if a peg exists between $\theta_1$ and $\theta_2$, ensuring that one $\theta_1$ can always be exchanged for one $\theta_2$. For simplicity, we ignore the total fee $F_{forward}$ that would be incurred to use the bridge, and state that given token prices at time $t$, denoted as $\textit{price}(\theta_1,t)$ and $\textit{price}(\theta_2,t)$, the bridge must ensure%
\begin{equation}\label{eq:token_equal2}
    v_1 \cdot \textit{price}(\theta_1, t) \equiv v_2 \cdot \textit{price}(\theta_2, t), \quad \forall t.
\end{equation}%
To uphold this peg, the bridge must maintain both liveness and security. If the bridge experiences technical failure or an attack, equivalency in \eqref{eq:token_equal2} may break down, leading to price divergence between the native and bridged tokens, which can disrupt user trust and capital efficiency.

\paragraph{Bridge Causality} A secure bridge must enforce causality, ensuring that no user extracts value that was never created or loses value that was not meant to be lost. Formally:

\begin{equation}\label{eq:causalityeq1}
    \begin{aligned}
         &\forall tx_{B_2}=(\theta_2, , , c_2, t_2), \quad \exists ! tx_{B_1}=(\theta_1, , , c_1, t_1) \\
         &\forall tx_{B_1}=(\theta_1, , , c_1, t_1), \quad \exists ! tx_{B_2}=(\theta_2, , , c_2, t_2) \\
         &\text{ such that } t_1 < t_2.
    \end{aligned}
\end{equation}

This ensures a bijective mapping between transactions on $b_1$ and $b_2$, preventing double-minting or loss of funds. Additionally, the bridge must preserve temporal ordering, enforcing $t_1 < t_2$ so that a transaction on $b_1$ must occur before its corresponding transaction on $b_2$.

\paragraph{Bridge Consistency} The bridge must guarantee that tokens locked on $b_1$ remain inaccessible until the corresponding tokens on $b_2$ are either burned or locked. This consistency constraint ensures that the system does not create or destroy value arbitrarily.

\begin{equation}\label{eq:crosschain1}
     v_2 \neq 0 \Rightarrow \neg \exists tx_{B_1}(\theta_2,v_1, , c_1, )
\end{equation}

Equation \eqref{eq:crosschain1} prevents a user from simultaneously withdrawing locked assets on $b_1$ and holding minted tokens on $b_2$, which could result in infinite money creation.

\begin{equation}\label{eq:crosschain2}
    v_2 > 0 \Rightarrow \exists ! \theta_1 \text{ such that } c_1 \mapsto  a_2.
\end{equation}

Equation \eqref{eq:crosschain2} strengthens this by enforcing a strict one-to-one correspondence between locked and minted tokens.

\paragraph{Value at Risk} Failures in bridge security can lead to asset loss with custodial, contract, and economic risks. In custodial risk, an externally controlled address or multi-sig wallet is compromised, and funds are permanently stolen. In contract risk, a smart contract-based bridge has a vulnerability (e.g., reentrancy attack) where attackers can extract tokens illicitly. In economic risk, a bridge relies on economic assumptions (e.g., optimistic rollups or bonded validators), but adversaries manipulate incentives to destabilize the system.

Losses may occur at i) the bridge contract level, where contract $c_1 $ on $b_1$ and contract $c_2$ on $b_2$ impact all users of the bridge, or ii) at the user level, where losses occur for individual user addresses $a_1$ or $a_2$. User-level losses (e.g., $a_2$ is not a valid address) are orthogonal to bridge security as long as they are not caused by flaws in bridge design.

A bridge’s trust model (custodian-based, validator-based, or smart contract-based) influences its risk profile and determines how funds are secured.

\subsection{Formalization of Bridge Attack Surfaces}
\label{sec:formaliseattacksurface}

Current research on bridge attacks is fragmented, lacking a consistent definition of attack surfaces and vectors. We present a formalized model for analyzing bridge security.

\begin{definition}[Attack]\label{attackdef}
    If \eqref{eq:token_equal2}, \eqref{eq:causalityeq1}, \eqref{eq:crosschain1}, or \eqref{eq:crosschain2} are violated due to malicious agents, the bridge is under attack.
\end{definition}

\begin{definition}[Failure]\label{failuredef}
    If \eqref{eq:token_equal2}, \eqref{eq:causalityeq1}, \eqref{eq:crosschain1}, or \eqref{eq:crosschain2} are violated due to technical errors but without malicious intent, the bridge is experiencing failure.
\end{definition}

\begin{theorem}
    If a bridge experiences an attack or failure, then \eqref{eq:token_equal2} is always violated.
\end{theorem}

Due to space limitations, the proof is given in Appendix~\ref{app:proof}.

\begin{definition}[Attack Vector $V$]\label{def:attack-vector}
    A vulnerability, pathway, or method that a malicious agent can exploit to launch an attack on a bridge $\mathbb{B}_{1 \leftrightarrow 2}$.
\end{definition}

\begin{definition}[Attack Surface $\Sigma $]\label{}
    The sum of all attack vectors for a given $\mathbb{B}_{1 \leftrightarrow 2}$. Formally, an attack surface is defined as $\Sigma = \langle \T, \mathcal{I} \rangle$, where $\T$ represents trusted entities and $\mathcal{I}$ represents implementation details.
\end{definition}

An attack surface can be decomposed into sub-surfaces for finer-grained analysis. A sub-surface $\Sigma^{'}_{\mathbb{B}_{1 \leftrightarrow 2}} \in \Sigma_{\mathbb{B}_{1 \leftrightarrow 2}} $ is defined as:

$$\Sigma^{'}_{\mathbb{B}_{1 \leftrightarrow 2}} = \langle \T^{'},\mathcal{I}^{'} \rangle, \quad \text{where} \quad \T^{'} \in \T, \quad \mathcal{I}^{'} \in \mathcal{I}.$$

Attack vectors in $\Sigma^{'}_{\mathbb{B}_{1 \leftrightarrow 2}}$ use all trusted entities $\T^{'}$ and span specific implementations $\mathcal{I}^{'}$.

\subsubsection{Examples of Attack Surfaces} In custodial key compromise ($\T = \{C\} $),   a centralized custodian’s private key is compromised, and all assets in the bridge are at risk. In the reentrancy attack on smart contract-based bridges ($\T = \{SCs\} $), an attacker exploits an improper withdrawal function to repeatedly drain funds. In validator collusion ($\T = \{V\} $), the validators of a PoS-based bridge collude; they approve fraudulent transactions and steal assets.

\subsubsection{An Attack Surface-Damage Model}
 
\begin{definition}[Damage Potential/Effort Ratio]
The Damage Potential/Effort Ratio of an attack vector $ V $, denoted $ der(V) $, quantifies the feasibility of an attack and is defined as $der(V) = \frac{I(V)}{E(V)}$, where $I(V)$ represents the expected damage or impact of the attack vector and $E(V)$ represents the computational, economic, or procedural effort required to execute the attack. 
\end{definition}

An attack vector $V$ is viable if and only if $der(V) > 1$, indicating that the expected damage outweighs the effort required to execute the attack. Hence, 
 
$$
der(V) = \begin{cases} 
    1, & \text{if } der(V) > 1 \\
    0, & \text{otherwise}
  \end{cases}
$$

We define the attack surface area of a bridge $ \mathbb{B}_{1 \leftrightarrow 2} $, or a subset of it, containing $ n $ attack vectors as:
\begin{equation}
    Area(\Sigma^{'}_{\mathbb{B}_{1 \leftrightarrow 2}}) = \sum_{i=1}^{n} der(V_i)
\end{equation}

We apply this framework to analyze security across distinct bridge implementation layers, including i) on-chain contracts (handling token locking, minting, and withdrawal), ii) off-chain relayers (notary or validator networks responsible for cross-chain verification) and iii) finality mechanisms (ensuring transaction irreversibility and bridging safety). This layered security model enables targeted risk mitigation by isolating attack surfaces based on implementation constraints.

\subsection{Attack Surfaces of Blockchain Bridges}

Blockchain bridges rely on multiple components (on-chain smart contracts, off-chain relayers, and destination chain mechanisms), each introducing distinct security risks. We formalize these risks through attack surface analysis and summarize them in Table~\ref{tab:attack_vectors_categorized}. Vectors are defined in Appendix~\ref{appsec:vectors}.

\subsubsection{Source Chain Layer Security} The attack surface of a bridge’s source chain implementation is given by $\mathcal{I}_{\textit{src}} = \langle \T_{\textit{src}}, R_{\textit{src}} \rangle$. For a bridge deploying $n$ smart contracts on the source chain, we model its attack surface as $\Sigma_{\textit{src}} = \langle \{SC_1, SC_2, \dots, SC_n\}, \T_{\textit{src}}\rangle$ (i.e., contract implementations and trusted entities), with an attack surface area $Area(\Sigma_{\textit{src}}) = \sum_{i=1}^{n} der(SC_i).$ For validator-controlled implementations (${SC_i} = \emptyset$), the attack surface consists solely of the trust set $\T_{\textit{src}}$ (i.e., the validators). While the contractual surface area is zero, the system remains exposed to validator-level risks such as economic incentives, collusion, or key compromise.
\subsubsection{Off-Chain Layer Security} The off-chain layer, which handles bridge communication, consists of notaries, light clients, or sidechains, each with unique trust assumptions as we discuss next.

\paragraph{Notary-Based Mechanisms}
For notary-based bridges ($\T_\textit{off} = \text{notaries}$), security improves with more notaries, but decentralization is costly. If a sufficient fraction of notaries turn malicious, they can control the bridge for a duration $t > t^*$, enabling asset theft.

\paragraph{Light Client Mechanisms}
For light-client bridges ($\T_\textit{off} = \{LC, MP\}$), security depends on Merkle proofs ($MP$). If an attacker controls the light client, they can alter transaction verifications. The attack surface reduces to
$Area(\Sigma_\textit{off} \mid \T_\textit{off} = \text{light client}) = der(\text{light client})$

\paragraph{Sidechain-Based Mechanisms} Sidechains use independent consensus rules (i.e., $R_{src}$). If $R_{src}$ is compromised, an attacker can mint arbitrary tokens or modify the bridge state. If the sidechain guarantees $R_{src}$ security, the attack surface vanishes $Area(\Sigma_\textit{off} | \T_\textit{off} = \text{sidechain}) = 0 \Leftrightarrow der(R_{src}) = 0.$

\subsubsection{Destination Chain Layer Security} The destination chain's attack surface depends on its reliance on smart contracts ($SC$) and custodians $\Sigma_\textit{dest} = \langle \tau_\textit{dest},\T_\textit{dest} \rangle, \quad \tau_\textit{dest} = \{SC_1, SC_2,\dots, SC_n\} \cup \{\textit{custodian}\}.$ Attack surface area is $ Area(\Sigma_\textit{dest}) = n + 1$, assuming a custodian is present.

\subsubsection{Total Attack Surface of a Blockchain Bridge} The total bridge attack surface is 
$\Sigma_{\mathbb{B}_{1 \leftrightarrow 2}} = \Sigma_\textit{src} \cup \Sigma_\textit{off} \cup \Sigma_\textit{dest} \cup \Sigma_\textit{other}.$ Other attack vectors include updates to bridge protocols, governance failures, or rug-pulls. We validate our framework by classifying past major bridge exploits under this model.

\begin{table}[h]
\centering
\caption{Categorized attack/disruption vectors across bridge layers. A dagger ($\dagger$) marks vulnerabilities that only arise when a bridge architecture actually employs the relevant component (for example, an oracle).}
\label{tab:attack_vectors_categorized}
\resizebox{1\columnwidth}{!}{%
\begin{tabular}{l c c c}
\toprule
\textbf{Attack Vector} & \textbf{Source Chain} & \textbf{Off-Chain} & \textbf{Destination Chain} \\
\midrule
\multicolumn{4}{l}{\textbf{Contract Logic \& Code Vulnerabilities}} \\
V1: Reentrancy attacks & \checkmark & & \checkmark \\
V2: Integer and arithmetic errors & \checkmark & & \checkmark \\
V3: Access control and forged account flaws & \checkmark & & \checkmark \\
V4: Race condition attacks & \checkmark & & \checkmark \\
V5: Unsafe external call exploits & \checkmark & & \checkmark \\
V6: Malicious event log manipulation & \checkmark & & \checkmark \\
V7: Contract upgrade risks & \checkmark & & \checkmark \\
\midrule
\multicolumn{4}{l}{\textbf{Authentication \& Authorization Failures}} \\
V8: Fake burn/lock proofs & \checkmark & \checkmark & \checkmark \\
V9: Malicious transaction modification & \checkmark & \checkmark & \checkmark \\
V10: Light-client verification flaws & \checkmark & \checkmark & \checkmark \\
V11: Oracle manipulation & \checkmark$\dagger$ & \checkmark & \checkmark$\dagger$ \\
V12: Malicious custodian manipulation & \checkmark$\dagger$ & \checkmark & \checkmark$\dagger$ \\
V13: Private key leakage or theft & \checkmark$\dagger$ & \checkmark$\dagger$ & \checkmark$\dagger$ \\
\midrule
\multicolumn{4}{l}{\textbf{Replay, Race, and Timing\textendash Based Attacks}} \\
V14: Timestamp manipulation & \checkmark & & \checkmark \\
V15: Replay attacks & & \checkmark & \checkmark \\
\midrule
\multicolumn{4}{l}{\textbf{Consensus \& Infrastructure Risks}} \\
V16: Consensus failure (51\% attack) & \checkmark & & \checkmark \\
V17: Delayed finality exploitation & \checkmark & & \checkmark \\
V18: Validator equivocation or misbehavior & \checkmark$\dagger$ & \checkmark &  \checkmark$\dagger$\\
V19: Denial of Service attacks & \checkmark & \checkmark & \checkmark \\
V20: Deep chain reorganization & \checkmark & & \checkmark \\
V21: Unbounded withdrawal limits &\checkmark & & \checkmark \\
V22: Rugpull & & \checkmark & \checkmark$\dagger$\\
\midrule
\multicolumn{4}{l}{\textbf{Front-End and Off-Chain Manipulation}} \\
V23: Front-end deception & & \checkmark & \\
\bottomrule
\end{tabular}}
\end{table}


\section{Bridge Design Patterns and Implementation Landscape}
\label{sec:bridge-summary}

To contextualize our formalism and threat model, we survey representative bridge implementations, analyzing their design across three key layers: source chain, off-chain coordination, and destination chain. Appendix Table~\ref{tab:bridge_implementations} summarizes their implementation patterns, trust assumptions, functional types, and blockchain coverage. Full protocol descriptions appear in Appendix~\ref{appendix:bridge-details} and the ecosystem is described in Appendix~\ref{appendix:state}.

\subsection{Static Analysis of Bridges}
\begin{table}[ht]
\caption{Access control and code structure metrics of bridge smart contracts}

\label{tab:accessControl}
\centering
\scriptsize
\resizebox{\linewidth}{!}{
\begin{tabular}{lrrrrr}
\toprule
\textbf{Bridge Name} & \textbf{Local vars} & \textbf{Inheritances} & \textbf{Modifier Count} & \textbf{RoleBased} & \textbf{Standard Libs} \\
\midrule
Avalanche BridgeToken         & 2  & 1 & 0 & Yes & 1 \\
Wormhole BridgeImplementation & 58 & 2 & 1 & No  & 3 \\
Arbitrum L1GatewayRouter      & 9  & 5 & 2 & No  & 0 \\
Arbitrum L1ERC20Gateway       & 3  & 1 & 1 & No  & 1 \\
Arbitrum L1CustomGateway      & 5  & 2 & 2 & No  & 1 \\
Arbitrum L1WethGateway        & 0  & 1 & 0 & No  & 2 \\
Stargate Router               & 30 & 3 & 1 & No  & 3 \\
DeBridge DeBridgeGate         & 44 & 5 & 4 & Yes & 1 \\
Across HubPool                & 50 & 5 & 3 & No  & 3 \\
Stargate Bridge               & 35 & 3 & 1 & No  & 1 \\
Allbridge Bridge              & 12 & 4 & 0 & No  & 1 \\
Nomad BridgeToken             & 6  & 4 & 0 & No  & 0 \\
Base L1StandardBridge         & 1  & 2 & 0 & No  & 0 \\
Optimism L1StandardBridge     & 0  & 2 & 0 & Yes & 0 \\
Hyperliquid Bridge2           & 66 & 2 & 0 & Yes & 5 \\
Meson BridgeV2                & 2  & 2 & 4 & Yes & 3 \\
\bottomrule
\end{tabular}}
\end{table}

\begin{table}[ht]
\caption{Lines of code (LOC), function visibility, and variable usage in bridge smart contracts}

\label{tab:codeAndFunctionMetrics}
\centering
\scriptsize
\resizebox{\linewidth}{!}{
\begin{tabular}{lrrrrrrr}
\toprule
\textbf{Bridge Name} & \textbf{LOC} & \textbf{Total lines} & \textbf{Public} & \textbf{External} & \textbf{Internal} & \textbf{Private} & \textbf{Global vars} \\
                     &             &                      & \textbf{Funcs}  & \textbf{Funcs}    & \textbf{Funcs}    & \textbf{Funcs}   & \textbf{Declared} \\
\midrule
Avalanche BridgeToken         & 155 & 226 & 9 & 0  & 0  & 1 & 6 \\
Wormhole BridgeImplementation & 630 & 776 & 19 & 3  & 14 & 0 & 0 \\
Arbitrum L1GatewayRouter      & 208 & 305 & 5 & 4  & 1  & 0 & 2 \\
Arbitrum L1ERC20Gateway       & 112 & 161 & 5 & 2  & 0  & 0 & 6 \\
Arbitrum L1CustomGateway      & 165 & 243 & 5 & 3  & 0  & 0 & 6 \\
Arbitrum L1WethGateway        & 78  & 92  & 2 & 1  & 4  & 0 & 2 \\
Stargate Router               & 280 & 323 & 0 & 17 & 4  & 0 & 5 \\
DeBridge DeBridgeGate         & 905 & 1110 & 7 & 23 & 15 & 0 & 27 \\
Across HubPool                & 610 & 1076 & 24 & 2  & 13 & 0 & 16 \\
Stargate Bridge               & 249 & 310 & 2 & 14 & 4  & 0 & 9 \\
Allbridge Bridge              & 215 & 320 & 1 & 11 & 2  & 0 & 5 \\
Nomad BridgeToken             & 150 & 245 & 7 & 5  & 0  & 0 & 6 \\
Base L1StandardBridge         & 260 & 321 & 8 & 8  & 6  & 0 & 2 \\
Optimism L1StandardBridge     & 220 & 324 & 3 & 9  & 6  & 1 & 3 \\
Hyperliquid Bridge2           & 570 & 840 & 17 & 15 & 2  & 14 & 18 \\
Meson BridgeV2                & 167 & 210 & 8 & 8  & 0  & 2 & 4 \\
\bottomrule
\end{tabular}}
\end{table}

\begin{table}[ht]
\caption{Analysis of external call behavior and defensive programming practices in bridge implementations}

\label{tab:callsAndErrors}
\centering
\scriptsize
\resizebox{\linewidth}{!}{
\begin{tabular}{lrrrrrrr}
\toprule
\textbf{Bridge Name} & \textbf{Ext Funcs} & \textbf{Low-level} & \textbf{Untrusted} & \textbf{Reentry Guard} & \textbf{Require/Assert} & \textbf{Custom Errors} & \textbf{Checks/Fn} \\
\midrule
Avalanche BridgeToken         & 5  & 0 & 0 & No  & 14 & 0  & 1.56 \\
Wormhole BridgeImplementation & 2  & 5 & 0 & Yes & 21 & 0  & 0.58 \\
Arbitrum L1GatewayRouter      & 3  & 0 & 0 & No  & 9  & 0  & 1 \\
Arbitrum L1ERC20Gateway       & 2  & 1 & 0 & Yes & 3  & 0  & 0.43 \\
Arbitrum L1CustomGateway      & 4  & 0 & 0 & Yes & 6  & 0  & 0.75 \\
Arbitrum L1WethGateway        & 0  & 0 & 0 & No  & 3  & 0  & 0.43 \\
Stargate Router               & 15 & 0 & 0 & No  & 8  & 0  & 0.38 \\
DeBridge DeBridgeGate         & 18 & 3 & 3 & Yes & 26 & 20 & 0.87 \\
Across HubPool                & 13 & 3 & 3 & Yes & 26 & 0  & 1 \\
Stargate Bridge               & 13 & 0 & 0 & No  & 8  & 0  & 0.5 \\
Allbridge Bridge              & 5  & 0 & 0 & No  & 11 & 0  & 0.9 \\
Nomad BridgeToken             & 4  & 0 & 0 & No  & 4  & 0  & 0.33 \\
Base L1StandardBridge         & 4  & 0 & 0 & No  & 4  & 0  & 0.25 \\
Optimism L1StandardBridge     & 4  & 0 & 0 & No  & 0  & 0  & 0 \\
Hyperliquid Bridge2           & 7  & 0 & 0 & Yes & 19 & 0  & 1 \\
Meson BridgeV2                & 7  & 1 & 1 & No  & 7  & 0  & 0.58 \\
\bottomrule
\end{tabular}
}
\end{table}

To assess the security robustness of bridge smart contracts, we conduct a \textit{static analysis} across key dimensions drawn from our layered attack surface model. Static analysis focuses on examining the contract's code without executing it, allowing us to detect vulnerabilities in logic, structure, and access control that may compromise the security priors \textit{bridge causality}, \textit{consistency}, and \textit{token peg integrity}. We analyze the bridges in terms of i) access control and code structure metrics  (Table~\ref{tab:accessControl}), ii) function visibility and variable usage (Table~\ref{tab:codeAndFunctionMetrics}) and iii) call behavior and defensive programming constructs (Table~\ref{tab:callsAndErrors}). This analysis maps to several high-impact attack vectors from Table~\ref{tab:attack_vectors_categorized}, notably: V1: reentrancy attacks, V3: access control and forged account flaws, V5: unsafe external call exploits, V7: contract upgrade and misconfiguration risks. We identify several notable patterns and anomalies that offer insights into the heterogeneous security postures and implementation philosophies across bridges.

One immediate observation in Table~\ref{tab:callsAndErrors} is the inconsistent use of reentrancy protection mechanisms. Despite having multiple externally callable functions, contracts such as \textit{Stargate Router} (15 external functions) and \textit{Stargate Bridge} (13 external functions) do not implement any form of reentrancy guard. In contrast, similarly sized contracts like \textit{DeBridgeGate} and \textit{Across HubPool} use such guards appropriately. The absence of reentrancy protection in these high-exposure contracts reflects a reliance on architectural assumptions. 

We also find a curious disconnect in Table~\ref{tab:accessControl} between the use of role-based access control and the deployment of Solidity modifiers. For example, \textit{Hyperliquid Bridge2} and Avalanche \textit{BridgeToken} both report role-based access (\textquote{RoleBased = Yes}) but have zero modifiers. This suggests that access restrictions may be implemented inline via \enquote{require} statements rather than modularized through modifiers, leading to less auditable and reusable control logic. In contrast, \textit{Nomad BridgeToken} and \textit{Allbridge} show zero modifier usage and lack role-based protection, increasing susceptibility to unauthorized access or logic misbehavior (V3). \textit{Nomad}'s exploit history confirms this: a faulty initialization allowed any user to spoof legitimate senders. 

A particularly unusual case in Table~\ref{tab:callsAndErrors} is the \textit{Optimism  L1Standard Bridge} contract, which lacks any require or assert statements despite being externally callable and declaring role-based control. While this absence would typically raise concerns (since most bridges include basic sanity checks), it reflects a deliberate design choice. Optimism and \textit{Base L1StandardBridge} rely on rollup-native architecture, where Ethereum layer-1 consensus enforces external validation. This structural safeguard reduces reliance on in-contract defensive coding against V3 and V5 vectors, but the lack of internal guards still poses risks if such contracts are reused outside this tightly scoped context.

In terms of architectural modularity, some contracts in Table~\ref{tab:codeAndFunctionMetrics}  show a preference for deeply internalized logic structures. \textit{Wormhole BridgeImplementation} and \textit{Hyperliquid Bridge2} contain large numbers of internal or private functions (13-15 each), indicating encapsulated logic. While this may reflect thoughtful separation of concerns, it also complicates external audits and transparency.

We further observe disproportionate uses of access control infrastructure relative to contract size in Tables~\ref{tab:accessControl} and \ref{tab:codeAndFunctionMetrics}. \textit{Meson BridgeV2}, with just 167 lines of code, defines four modifiers and a complete role-based access layer. In contrast, \textit{Stargate Router} spans over 280 lines but employs one modifier. Such disparity shows divergent security philosophies among bridge developers, with some opting for granular control at all costs and others prioritizing operational simplicity.

A notable structural pattern in Table~\ref{tab:codeAndFunctionMetrics} appears in the allocation of global versus local variables. \textit{DeBridgeGate} maintains a high number of global state variables (27), reflecting rich on-chain logic and persistent data tracking. Conversely, \textit{Wormhole Bridge} holds 58 local variables but defines no global state, potentially due to reliance on proxy patterns or off-chain state.

Low-level calls (\texttt{call}, \texttt{delegatecall}, \texttt{staticcall}) introduce attack surface via vector V5. \textit{Wormhole}, \textit{DeBridge}, and \textit{Across} use them multiple times, sometimes in the presence of untrusted inputs. While \texttt{call} can be necessary (e.g., for token forwarding), its misuse or failure to handle return values correctly has led to major incidents such as the \textit{Poly Network} and \textit{Qubit} \cite{qubit_finance} exploits. Bridges like \textit{Avalanche BridgeToken}, \textit{Nomad}, and \textit{Allbridge} avoid low-level calls entirely, reducing exposure to unsafe execution paths.

Lastly, we compute the average number of checks per function (\#Checks/\#Functions) as a coarse proxy for defensive programming density. Contracts such as \textit{Avalanche BridgeToken} (1.56) and \textit{Hyperliquid Bridge2} (1.00) stand out as aggressively guarded, while \textit{Optimism L1StandardBridge} (0.00) again reflects an absence of such measures. These findings indicate varying degrees of maturity in secure smart contract engineering practices.

Static analysis reveals that bridge contracts vary widely in their defensive quality.  \textit{DeBridgeGate} sets a high bar for smart contract hygiene, leveraging modifiers, role-based controls, and comprehensive assertions to guard its complex logic. \textit{Wormhole} and \textit{Across} balance complexity with modular defenses but remain susceptible to implementation flaws, as demonstrated in \textit{Wormhole}’s historical signature validation failure. \textit{Stargate}, \textit{Nomad}, and \textit{Allbridge} show critical gaps in protective constructs, leaving them exposed to known attack vectors. These are particularly concerning given their public deployment.

\subsection{Transaction Analysis of Bridges}

\begin{figure}[h]
    \centering
\includegraphics[width=0.9\linewidth]{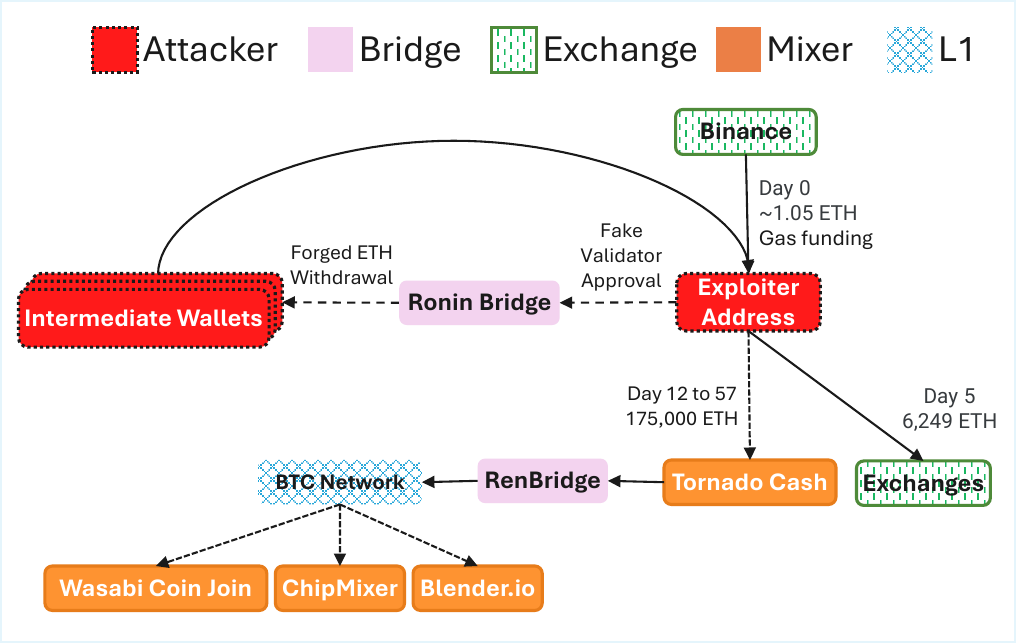}
\caption{
Attack flow of the \textit{Ronin Bridge} exploit. The attacker gained control of 5 out of 9 Ronin Bridge validators and used them to sign and submit forged withdrawals. The stolen funds amounted to 173,600 ETH, along with an additional 25.5 million USDC that was exchanged for 8,564 ETH.}

    \label{fig:ronin-attack}
\end{figure}

To better understand the on-chain behavior of bridge exploiters, we perform a transaction-level analysis focused on Ethereum addresses used by attackers, as Ethereum is commonly chosen for fund collection due to its exchange access and availability of mixing services. We obtain Ethereum transaction data by running a full node using Geth (\url{https://github.com/ethereum/go-ethereum}) and parsing all transactions via Ethereum-ETL (\url{https://github.com/blockchain-etl/ethereum-etl}). We visualize the attack patterns in transaction graphs to capture attacker behavior. The analysis includes both externally owned account transactions and internal transactions, covering pre-attack, attack, and post-attack phases, where the phase lengths span the full Ethereum history up to March 2025.

Due to space limitations, an example transaction subgraph is shown in \autoref {fig:ronin-attack}, and seven other notable attacks are visualized in \autoref{sec:realAttacks}. In the figure, the \textit{Ronin bridge} to Ethereum was victim to a validator key compromise in an attack now attributed to North Korea’s Lazarus Group. \textit{Ronin} used a 5-of-9 multisig validation for bridge withdrawals. Attackers compromised five private keys through social engineering. Once in control, they issued two fraudulent transactions, draining 173,600 ETH and 25.5M USDC (worth approximately \$624 million at the time) from the \textit{Ronin bridge} in a single stroke. This event, the largest DeFi hack ever, was essentially a failure of the bridge’s trust model: the off-chain validators were assumed honest, but the minimal quorum and centralized key management (a single entity controlled 4 of 9 validators) made it easy for an attacker to breach. The breach went unnoticed for six days until a user discovered it.

Most exploiter addresses have little to no previous on-chain activity. In many cases, the wallets were newly created and funded with just a sufficient amount to cover gas fees. The average number of pre-attack transactions was 28, suggesting that even the more \textquote{active} wallets were only lightly used. Some of these initial funds even originated from mixers, most notably Tornado Cash \cite{tornado_cash} in the \textit{Wormhole} bridge hack, where the exploiter received ETH from Tornado Cash before launching the attack, an indication of deliberate obfuscation of the funding source.  In other instances, such as the \textit{Nomad} and \textit{Horizon bridge} \cite{horizon_bridge} hacks, attackers used unlabeled intermediate wallets, further complicating traceability, and as a strategy to hide the funding sources.

During the attack period, the behavior is highly focused. The attacker often interacts directly with the bridge contract via a few function calls, most of the time, only one forged withdrawal or proof verification. In attacks such as \textit{Nomad}, we observe many small transactions from multiple addresses exploiting the same vulnerability after it became public, showing replication once the vulnerability was made public.

Internal transactions are important for creating the contract call sequences. The contract call traces highlight how the attack is carried out at the smart contract level, often using call stacks involving specific functions like \textit{withdrawERC20For()} or \textit{verifyProof()}. This helps verify the exploit strategies and confirms the role of compromises or logic flaws.

After the exploit, the stolen funds are moved rapidly to avoid detection. Some exploiters transferred stolen funds through multiple intermediate wallets before returning them to a central exploiter address, for instance, in the \textit{Ronin Bridge} hack. Some others bridged funds to other blockchains such as Ethereum, Polygon, and Avalanche using bridges. Many proceeds are laundered through privacy tools like Tornado Cash or converted via decentralized exchanges (DEXs).

After initial laundering, several attacker addresses remain active. In our analysis, 6 out of 18 wallets continued transacting for more than one year after the exploit, indicating possible long-term usage or reactivation of compromised wallets.

Although certain early detection remains challenging, our analysis identifies behavioral clues that may serve as early warning signs of an upcoming exploit. A common pattern is the initial funding of a wallet with just enough ETH to cover gas fees, which is sourced from mixers. These wallets usually indicate little to no prior activity. Moreover, in order to verify system reactions or contract behavior, numerous low-value probing transactions are occasionally issued prior to the main exploit. These patterns could be leveraged to flag suspicious activity in near-real time.

\section{Security Analysis of Bridge Attacks}
Our analysis reveals that bridge causality and consistency priors are violated in bridge attacks. Most attacks involve violations of the cross-chain causality prior. While Equation~\eqref{eq:token_equal2}, the peg prior, is formally violated as a result of these attacks, the peg itself is not the target of the attacks. That is, attackers do not manipulate token prices or attempt to create price divergence across chains. Instead, the violations arise indirectly because tokens are minted or released without proper backing, breaking the assumption of value parity that the peg prior encodes. Due to space limitations, we describe the attacks in Appendix~\ref{sec:realAttacks}, and show a meta-analysis in Table~\ref{tab:bridge_attacks}.

\subsubsection*{Attack Vector Analysis.} \autoref{tab:bridge_attacks} shows that V3 (access control) and V13 (key leakage) have been exploited 10 times each. The exploits reveal two broad categories of failures: (1) Off-chain trust failures and (2) On-chain validation failures. The off-chain trust failures encompass all incidents where the bridge’s security relied on individuals or off-chain systems that were compromised, notably the multisig key compromises (\textit{Ronin}, \textit{Harmony}, \textit{Multichain}, \textit{Orbit}) and related cases. In \textit{Ronin} and \textit{Harmony}, the sources of failure were outside the blockchain: hackers penetrated the organizations controlling the validators and obtained the private keys needed to sign fake transactions. These attacks did not exploit a bug in code; they exploited insufficient decentralization and operational security. Essentially, the assumption that a small set of validators would remain honest was violated. When 5 of 9 \textit{Ronin} validators and 2 of 5 \textit{Harmony} signers turned malicious (via key theft), the bridge smart contracts on-chain duly obeyed the malicious signatures, an example of how improperly validated inputs led the system to execute unintended behavior. Similarly, \textit{Multichain}’s collapse was an off-chain failure: the system’s architecture secretly concentrated too much trust in one individual. These illustrate that trusted or trust-minimized bridges are vulnerable by design. They introduce new trust points (keys, signers, servers) that attackers target through phishing, insider collusion, malware and more. The on-chain validation failures, on the other hand, cover exploits like \textit{Poly Network}, \textit{Wormhole}, \textit{Nomad}, \textit{BSC Token Hub}, \textit{Qubit}, where the bridge smart contracts or crypto verification routines on one of the chains had a flaw. In these cases, the attacker manipulated the code logic (e.g., forging a message that the contract erroneously accepted as valid). For \textit{Poly}, the bug was an unchecked external call in a privileged contract; for \textit{Wormhole}, a bypassed signature verification on Solana; for \textit{Nomad}, an incorrect initialization setting that trusted everyone; for BSC, a flawed light-client proof verification.

Despite different mechanics, these represent software/security bugs in the bridge implementation. The failures occurred either on the source chain contract (for \textit{Poly} on Ethereum) or the destination chain contract (for \textit{Wormhole}, \textit{Nomad}, \textit{Qubit}) or the intermediate relay logic (as in \textit{BSC}’s light client). Crucially, these attack vectors map to points of validation in the bridge architecture.  

In several cases (\textit{Wormhole}, \textit{Nomad}, \textit{BSC}), the violated core assumption was that the smart contract correctly validates the cross-chain proof. These were not fundamental cryptographic failures (algorithms such as \textit{secp256k1} were sound), but errors in how the algorithms were applied. This suggests many bridge exploits are avoidable with more rigorous software engineering, such as comprehensive audits, formal verification of bridge contracts, and in-depth defense (e.g., requiring multiple independent checks of a proof). 

Importantly, many of these incidents exhibit high damage-to-effort ratios. The impact of a successful exploit has often reached hundreds of millions of dollars, while the effort required, whether key theft, phishing, or exploiting a poorly audited contract, remains relatively low. These high $\text{der}(V)$ values underscore why such attack vectors persist across different bridge designs and why they dominate the historical incident landscape.

\subsubsection*{RQ1: Are bridge attacks structurally preventable, or can they be mitigated by design?}  Yes, the evidence suggests that certain designs are far more robust. Notably, \textbf{no major exploits have occurred on fully trustless light-client bridges or rollup bridges}. For example, Cosmos’s L1-to-L1 IBC channels have operated without incident between dozens of chains since 2021 (a lossless reentrancy bug was discovered early~\cite{lindrea2024cosmos}), and Ethereum’s (young) rollup L1-to-L2 bridges have not been breached to date. Avalanche L1-to-L1 bridge to Ethereum continues functioning without any incidents. These systems benefit from minimal attack surface: there are no external signers to target, and the verification logic is often simpler or inherited from consensus (IBC uses tendermint-like light clients, which are well-vetted; rollups use post validity/fraud proofs checked by L1). This suggests that trustless designs inherently eliminate entire classes of attacks (notably, attack vector V3 of Table~\ref{tab:attack_vectors_categorized}).  These architectures inherently exclude common attack vectors like compromised keys or poorly controlled access mechanisms, offering significantly lower $\text{der}(V)$ values.

By contrast, most attacks struck bridges that added a new layer of validation on top of the two blockchains. Each additional component (e.g., a multisig, an off-chain oracle network, or a novel smart contract) introduced complexity and potential weakness. Our analysis shows evidence that the more a bridge can lean on intrinsic blockchain security (native verification), the safer it will be. In contrast, trust-minimized L1-to-L1 bridges are often only as strong as their validator set, which could be much weaker than either chain’s consensus. A quantified framework by LI.FI found in 2022 that many so-called trust-minimized bridges still effectively rely on a majority of a few validators, leaving them vulnerable to 51\% attacks or key compromises~\cite{ChandMurdock2022Spectrum}. Indeed, \textit{Ronin} and \textit{Harmony} hacks demonstrated that a <10-member validator set offers poor resilience. Increasing the set (e.g., Axelar has ~70+ validators) improves security, but unless those validators are as robust as a full blockchain consensus (with hundreds of nodes and economic security), the bridge remains a tempting honeypot for hackers. Therefore, trust-minimized bridges offer better resilience than naive trusted models, but they still concentrate trust among a group.

Properly decentralized networks with bonding (to slash malicious validators) are harder to attack; for instance, a hack of the scale of \textit{Ronin} on a larger bonded validator bridge has not occurred yet. However,  Ethereum-to-Solana (L1-to-L1) \textit{Wormhole}’s Feb 2022 case reminds us that even absent collusion, a bug in the validation code can be just as catastrophic. This blurs the line between \textit{fundamental} and \textit{implementation} issues: we argue that the complexity of blockchain bridge protocols is a fundamental weakness, as more complex logic yields a higher chance of errors.  

Finally, the prevalence of hacks has spurred research into failure-resistant bridges. For example, designs like \textquote{circuit breaker} bridges propose that if an unusually large withdrawal is attempted, the bridge only allows it after a delay or community vote, to mitigate instant draining. Others suggest blending multiple validation mechanisms, e.g., requiring both an off-chain multi-sig and an on-chain light client to agree, which could guard against either one failing alone (at the cost of added complexity). From a benchmarking perspective, we argue that diversity in validation (multi-layer security) might drastically reduce the probability of a successful exploit, albeit with performance trade-offs. Moreover, as bridges adopt techniques like zero-knowledge (zk) proofs for validation (e.g., \textit{zkBridge} frameworks), some traditional attack vectors might be closed (a zk proof can attest to a source chain state without any trusted signer at all). Such bridges could offer the holy grail: trustless interoperability between any two chains, with cryptographic guarantees and no reliance on human validators. Early prototypes (like \textit{zkRelay} and light-client circuits) are promising, but they require heavy computation and are just entering practical deployment. 

\subsubsection*{RQ2: Are bridge attacks detectable in real-time?} A surprising inability of the blockchain ecosystem is the lack of a strong analytics foundation that could monitor bridges for potential attacks. In the \textit{Ronin} hack, it took days (and a user complaint) to notice that \$620M had vanished. An obvious insight here is that real-time auditing and alerts are essential, but it raises a more fundamental question: why have such systems not been developed, despite the substantial capital invested in decentralized finance? 

We outline four reasons. First, many DeFi projects operate with lean teams focused primarily on product development and protocol maintenance, with their expertise originating from smart contract development. Allocating resources to build and maintain sophisticated monitoring systems often falls outside their immediate capabilities. Second, implementing real-time detection systems requires a scalable analytics infrastructure capable of handling large volumes of blockchain data. This necessitates expertise in data engineering and analytics, skills that may not be prevalent among smart contract developers. Third, analyzing blockchain data presents unique challenges, including model scalability and the accuracy of anomaly detection.  Traditional data analytics methods are often ill-suited to address the decentralized and heterogeneous structure of blockchain networks. Moreover, only a limited number of research efforts have focused on foundational questions specific to blockchain, such as identifying influential addresses within transaction networks (e.g., Alphacore~\cite{victor2021alphacore}; see~\cite{azad2024machine} for a recent survey).
 Lastly, investing in detection systems does not directly generate revenue, making it less appealing for projects to allocate funds toward such initiatives. While white-hat hacking offers some incentives, it often involves high effort with uncertain rewards. Nonetheless, advancements in automated security testing, particularly fuzzing~\cite{fuzzing}, are emerging as promising solutions. Fuzzing involves providing invalid, unexpected, or random data inputs to smart contracts to uncover vulnerabilities. 

The anomaly detection effort must eventually fall on bridge managers. Some bridges (e.g., \textit{Gravity}, \textit{Chainlink CCIP} \cite{chainlinkCCIP2025}) are now deploying independent, on-chain monitoring tools that can pause a bridge if suspicious activity is detected (for example, if an invariant is violated or if an unrecognized validator signature appears). In centralized finance, analogous systems flag large transfers for manual review; bridges could implement similar safeguards or circuit-breakers, but this reintroduces trust if not done in a decentralized way. 

\subsubsection*{RQ3: What mechanisms exist to mitigate damage and minimize asset loss?} There is no standard set of mechanisms for limiting damage once a bridge is under attack, and mitigation strategies remain ad hoc and uneven across protocols. \textit{Gravity} and \textit{Chainlink} are developing anomaly detection tools that can halt bridge operations in the face of an attack. In Celer, there is a buffer delay period during which transactions can be independently verified before a transfer is finalized. If inconsistencies are detected, the transfer can be halted, offering an added layer of fraud prevention. Post-factum, blockchain transparency does aid forensics; after these hacks, investigators (e.g., \textit{ZachXBT} on Twitter) traced funds through addresses and often identified suspects or recovered portions. But prevention and rapid response are clearly preferable to relying on clawbacks. Benchmark evaluations must incorporate attack detection latency as a metric; i.e., how quickly a breach is noticed and halted. Shorter detection times (as in BSC’s case, hours) can significantly limit losses compared to delayed discovery (as in \textit{Ronin}).

The bridge attacks of 2021–2025 underscore a core tension in blockchain interoperability: security vs. connectivity. Bridges expand what users can do (moving assets across chains), but they also expand the attack surface beyond any single blockchain. Our survey of exploits indicates that while many were due to low-hanging fruit bugs or a lack of operational security, there are also intrinsic risks whenever custody of funds shifts to a separate system. 

\subsubsection*{RQ4: What is the long-term outlook for bridge security? What foundational steps are needed to design formally verifiable or provably secure bridge protocols?} If truly trustless (or provably secure) bridge mechanisms can be standardized, we expect far fewer catastrophic failures. Until then, however, bridge designers must assume that attacks are a matter of when, not if, and design accordingly. This means employing rigorous audits, fail-safes, decentralized validation, and continuous monitoring as part of any cross-chain system. Past events suggest that increasing decentralization and using established, simple verification methods (light clients, chain consensus) correlate with higher security, whereas highly complex or centralized bridges have correlated with the largest failures. 

On the other hand, the multisig key hacks raise the question of whether those bridges were destined to fail, i.e., was it an inevitability that eventually some signer’s key would get compromised? One can argue that any valuable honeypot (whether a contract with billions of locked value or a key controlling billions) will, over time, face relentless attack until a weakness is found. Even with multi-party validators, if the set is small or the keys are not protected by robust hardware security modules and operational security, attackers have a focal point to strike. Thus, fundamental architecture choices (like using a 2-of-5 multisig with hot keys) can be seen as structural vulnerabilities, not in the theoretical sense, but in practical terms of presenting an irresistibly weak link (humans with keys) compared to the surrounding blockchain’s security.

A lesson learned with pain is that trusted bridges must be avoided due to significant management risks~\cite{kiepuszewski2022rollups}. Rollups offer a trustless solution but are limited to chains with compatible runtimes. New protocols like the \textit{IBC} protocol of Cosmos and \textit{Chainlink}'s Cross-Chain Interoperability Protocol aim to bridge heterogeneous chains. However, even those new-generation solutions introduce varying trust assumptions. IBC uses light clients when both chains natively support it. However, when adapted for non-IBC-compatible chains, additional trust assumptions are required. CCIP, on the other hand, relies on Chainlink's decentralized oracle network to facilitate cross-chain interactions, introducing a trust-minimized model where users depend on the honesty and reliability of the oracle nodes. 

\section{Conclusion}
\label{sec:conclusion}

Blockchain bridges have become indispensable infrastructure for cross-chain interoperability, yet they remain among the most vulnerable components of the decentralized ecosystem. Our study presents the first comprehensive, data-driven systematization of bridge security, combining formal modeling, large-scale static analysis, and empirical transaction-level investigations. We show that the majority of bridge attacks violate core security priors, particularly causality and consistency, without breaching the value peg itself. 

Our analysis reveals that most successful exploits stem from two dominant classes: off-chain trust failures and on-chain validation bugs. These vectors frequently exhibit high damage-to-effort ratios, which we formalize through the $\text{der}(V)$ metric. We also find that architectural design plays a decisive role in resilience. Trustless models, such as light-client and rollup-native bridges, have so far withstood real-world adversarial conditions, while trusted and loosely trust-minimized bridges remain disproportionately vulnerable. Emerging defense mechanisms like circuit breakers, buffer delays, and hybrid validation schemes offer promise but are inconsistently implemented and lack formal standardization.

Finally, we identify critical gaps in real-time detection and mitigation. Despite billions in locked value, most bridges lack robust monitoring infrastructure or containment protocols, and responses to attacks are often delayed and improvised. Bridging this gap will require new research into on-chain anomaly detection, decentralized fail-safe mechanisms, and benchmarking frameworks that account for detection latency and systemic risk.

We hope this work provides a foundation for rethinking how cross-chain systems are built, evaluated, and secured. As bridges continue to evolve, their long-term viability will depend not only on throughput and composability but also on their ability to withstand failure in adversarial conditions. Bridging across blockchains should not mean bridging across trust assumptions.

\bibliography{main}
\bibliographystyle{ieeetr}

\newpage
\section*{Appendix}

\appendix
\section{Related Work on Scalability, Interoperability and Bridge Designs}
\label{sec:extendedRelated}

Due to space constraints in the main text, this appendix outlines three additional research areas relevant to our work. We will start this section by analyzing scalability and interoperability studies.

\subsection{Blockchain Scalability and Off-Chain Protocols}

Recent bridge proposals aim to shift trust off-chain and reduce exposure by using succinct proofs, validator diversity, or formal verification. The literature reflects an active effort to reconcile decentralization, trust minimization, and performance in cross-chain systems.  For example, McCorry et al.~\cite{McCorry2021Sok:Blockchains} and Kim et al.~\cite{HangukTongsinHakhoe2018ICTCKorea} explore off-chain protocols and categorize scalability solutions, respectively.   Blockchain interoperability, the ability for different blockchain systems to interact, is a core focus of Lee et al.~\cite{Lee2023SoK:Hacks}, Wang~\cite{Wang2023ExploringSurvey}, Belchior et al.~\cite{Belchior2022ATrends}, Zamyatin et al.~\cite{Zamyatin2019SoK:Ledgers}, Qasse et al.~\cite{Qasse2019InterSurvey}, Monika et al.~\cite{Niranjan2020ProceedingsConference}, Schulte et al.~\cite{Schulte2019TowardsInteroperability}, Hardjono et al.~\cite{Hardjono2020TowardSystems}, Lafourcade et al.~\cite{Lafourcade2020AboutInteroperability}, and Duan et al.~\cite{Duan2023AttacksSurvey}. These studies highlight the challenges in ensuring ACID properties across diverse blockchains~\cite{Wang2023ExploringSurvey}, the limitations due to blockchain isolation~\cite{Qasse2019InterSurvey}, and the lack of interaction between different ledgers and legacy systems~\cite{Niranjan2020ProceedingsConference}. Methods vary: Lee et al.~\cite{Lee2023SoK:Hacks} focus on bridges, Zamyatin et al.~\cite{Zamyatin2019SoK:Ledgers} on communication protocols, Schulte et al.~\cite{Schulte2019TowardsInteroperability} suggest transitioning from closed to open systems, while Hardjono et al.~\cite{Hardjono2020TowardSystems} propose standardized architectures. Lafourcade et al.~\cite{Lafourcade2020AboutInteroperability} advocate for a unified blockchain.

\subsection{Security and Privacy Concerns in Blockchain Interoperability}

Security risks inherent to cross-chain protocols are discussed in Lee et al.~\cite{Lee2023SoK:Hacks}, Zamyatin et al.~\cite{Zamyatin2019SoK:Ledgers}, and Duan et al.~\cite{Duan2023AttacksSurvey}. Lee et al. examine bridge-specific vulnerabilities, Zamyatin et al. propose a trust-evaluation framework, and Duan et al. categorize systemic attack vectors. Borkowski et al.~\cite{Borkowski2018TowardsTAST} provide a comprehensive review of cross-blockchain technologies through the TAST research project. Caldarelli~\cite{Caldarelli2022WrappingTokens} introduces wrapped tokens for interoperability, though with reintroduced trust assumptions. Borkowski et al.~\cite{Borkowski2019Cross-BlockchainOutlook} explore atomic cross-chain transfers as a means to link otherwise isolated systems.

\subsection{Protocol Design, Relays, and Bridges}

Pioneering work by Herlihy et al.~\cite{Herlihy2019} introduced the notion of cross-chain deals, formulating atomic cross-chain transactions under adversarial conditions. Their protocol ensures that mutually distrusting parties exchanging assets across separate blockchains either all succeed or all abort, without a trusted intermediary. Atomic cross-chain swaps and hashed timelock contracts emerged from this foundation~\cite{Herlihy2019}, with further developments such as Anonymous Multi-Hop Locks~\cite{Malavolta2019} and Universal Atomic Swaps~\cite{Thyagarajan2022} generalizing these mechanisms for broader interoperability.

Relay and light-client protocols enable cross-chain state verification. BTC Relay provided a working, though costly, Ethereum-based verifier for Bitcoin transactions. Efficiency enhancements such as FlyClient~\cite{Bunz2020} and NiPoPoW~\cite{Kiayias2020,Scaffino2025} offer logarithmic communication overhead and succinct proofs. \textit{zkBridge} demonstrates state transition verification via \textit{zk-SNARKs}~\cite{Xie2022}, showcasing low-cost, trust-minimized bridging across Ethereum and Cosmos.

Asset-based bridges also remain prominent. XCLAIM~\cite{Zamyatin2019} introduced trustless, collateral-backed cross-chain assets. While projects like \textit{Wormhole} and \textit{PolyNetwork} use validator-based event notaries, they pose risks if quorum assumptions fail. These are complemented by formal frameworks such as \textit{zkRelay} and \textit{ETH Relay}, and conceptual models like Bitcontracts for Bitcoin–Ethereum smart contract interoperability.

Security of these bridges is increasingly scrutinized. SmartAxe~\cite{Liao2024} detects Cross-Chain Vulnerabilities through fine-grained static analysis. Attacks like the \$320M \textit{Wormhole} and \$610M \textit{PolyNetwork} incidents illustrate the cost of flawed validation and inadequate trust assumptions. Alba~\cite{Scaffino2025} proposes \textit{Pay2Chain} bridges that combine off-chain payment guarantees with cryptographic proofs to reduce bridge attack surfaces.

\begin{table*}[t]
\caption{Notations for Bridge Mechanisms}
\label{tab:notation}
\renewcommand{\arraystretch}{1.2}
\begin{tabular*}{\textwidth}{@{\extracolsep{\fill}} l l}
\toprule
\textbf{Notation} & \textbf{Description} \\
\midrule
$b$ & A blockchain. \\
$N_p$ & Set of peer-to-peer network nodes in a blockchain. \\
$N_a$ & Set of address nodes in a blockchain. \\
$\mathbb{H}$ & Chronological history of transactions in a blockchain. \\
$R$ & Consensus rules governing transaction creation in a blockchain. \\
$\theta$ & A token, with state defined by its transaction history $\mathbb{H}_{\theta}$. \\
$t$ & A point in time. \\
$tx$ & A transaction, defined as $tx = (\theta, v, a_1, a_2, t_{tx})$, transferring value $v$ of token $\theta$ from address $a_1$ to $a_2$ at time $t_{tx}$. \\
$a$ & A user address. \\
$b_1, b_2$ & Source and destination blockchains in a bridge. \\
$\mathcal{I}$ & Implementation mechanism of a bridge. \\
$\mathbb{B}_{1 \leftrightarrow 2}$ & A blockchain bridge between $b_1$ and $b_2$, defined as $\mathbb{B}_{1 \leftrightarrow 2} = \{b_1, b_2, \mathcal{I}\}$. \\
$\chi_t$ & The state of a cross-chain transaction at time $t$. \\
$TS$ & A node trust set. \\
$c$ & A smart contract address. \\
$v$ & Quantity of a token $\theta$. \\
$a \mapsto v$ & An address $a$ holds value $v$ of a token. \\
$F_{\text{forward}}$ & Bridge fee for forward transfer, defined as $F_{\text{forward}} = f_1 + f_2 + f^*$. \\
$F_{\text{reverse}}$ & Bridge fee for reverse transfer. \\
$f_1, f_2$ & Transaction processing fees on $b_1$ and $b_2$, respectively. \\
$f^*$ & Bridge operation fee. \\
$d$ & Duration associated with a blockchain operation. \\
$\mathcal{T}_\text{off}$ & Off-chain implementation. \\
$d_{b_x}$ & Block confirmation durations for $b_x$. \\
$d^*$ & Time for off-chain mechanism to process and signal a transfer. \\
$\T$ & Bridge trust set. \\
$\T_{\textit{src}}, \T_{\textit{dest}}$ & Trust assumptions related to the source and destination blockchains. \\
$\T_{\textit{off}}$ & Trust assumptions in off-chain mechanisms. \\
$L$ & Light client verifier. \\
$M$ & Merkle proof. \\
$N$ & Notary validator set. \\
$F$ & Total bridge fee for a transaction. \\
$D$ & Delay introduced in a bridge due to proof. \\
$Q$ & Quantity of token being referenced. \\
$MP$ & Message propagation. \\
$V$ & Attack vector. \\
$\mathbb{E}(C)$ & Expected loss or cost. \\
$\mathbb{E}(L)$ & Expected value at risk in a $CCB$. \\
$SCs$ & Set of smart contracts. \\
\bottomrule
\end{tabular*}
\end{table*}

\section{Proof of the Failure Theorem}
\label{app:proof}
\begin{theorem}
If a bridge experiences an attack or failure, then Equation~\eqref{eq:token_equal2} is always violated.
\end{theorem}

\begin{proof}
Equation~\eqref{eq:token_equal2} formalizes token parity across two blockchains $b_1$ and $b_2$, requiring that the locked value on $b_1$ matches the released or minted value on $b_2$:
\[
v_1 \cdot \textit{price}(\theta_1, t) \equiv v_2 \cdot \textit{price}(\theta_2, t), \quad \forall t.
\]

Suppose a bridge undergoes an attack or failure. By Definitions~\ref{attackdef} and~\ref{failuredef}, this implies that at least one of Equations~\eqref{eq:causalityeq1}, \eqref{eq:crosschain1}, or \eqref{eq:crosschain2} is violated.

\textit{Case 1: Violation of causality (\eqref{eq:causalityeq1}).}  
If a token transfer occurs on $b_2$ without a corresponding event on $b_1$, or vice versa, then either tokens are created without backing ($v_2 > v_1$), or locked value is unclaimed ($v_1 > v_2$). In both cases, the equivalence in~\eqref{eq:token_equal2} is broken.

\textit{Case 2: Violation of consistency (\eqref{eq:crosschain1} or \eqref{eq:crosschain2}).}  
If a user holds tokens on $b_2$ while also being able to access funds on $b_1$, or if the one-to-one correspondence between locked and minted tokens is violated, then value duplication or loss occurs. Again, this leads to $v_1 \cdot \textit{price}(\theta_1, t) \not\equiv v_2 \cdot \textit{price}(\theta_2, t)$.

In all such cases, the bridge no longer preserves token parity, hence Equation~\eqref{eq:token_equal2} is violated.

\end{proof}

\section{Cross-Chain Bridge Solutions}
\label{appendix:bridge-details}
 
We will discuss several popular cross-chain bridge protocols, cross-referencing them with their formalization. Specifically, we examine i) the implementation choices made by each bridge at different layers, ii) the trust assumptions underlying each bridge, iii) fee and time considerations, including comments on security where relevant, iv) the specific blockchains that each bridge connects. Table~\ref{tab:bridge_implementations} provides an overview of cross-chain bridge implementations and architectural components.

\subsection{Arbitrum Bridge}

Arbitrum is not a bridge itself, but a layer two rollup solution designed to prepare large batches of transactions off-chain, significantly reducing gas costs and scaling Ethereum. However, Arbitrum also includes a native token bridge, which lets users move assets between Ethereum (L1) and Arbitrum (L2). This is commonly referred to as the Arbitrum Bridge, and it's used to deposit and withdraw tokens like ETH or ERC-20s. 

As an optimistic rollup, Arbitrum assumes transaction validity by default and relies on economic incentives to detect and correct fraud. A centralized sequencer currently assembles and orders transaction batches, producing a new Arbitrum block for each Ethereum block. These transactions are compressed and posted as calldata (read-only, byte-addressable region) on Ethereum.  

A set of smart contracts deployed on Ethereum manages the core protocol logic, including the rollup mechanism, cross-chain messaging, fraud proofs, bridging, and third-party gateway interactions. The trust set is defined as $\tau = \{\text{SCs}\}$, representing these deployed smart contracts, which collectively enforce the security and correctness of the rollup.

\textit{Arbitrum’s bridge} is structured around a gateway bridge contract that delegates to more specific contracts (e.g., ERC20, WETH, or custom contracts). These specialized contracts are responsible for minting the corresponding tokens on layer 2. This modular design enhances security and extensibility by isolating token-specific logic across separate contracts.

To guard against fraud, Arbitrum implements a challenge-response protocol during a designated challenge period. Off-chain participants can dispute the validity of a rollup block through an interactive verification process often described as the \textquote{grasshopper game}, where the correct and fraudulent segments of execution are iteratively narrowed down. If fraud is detected, the protocol rolls back to the last valid state and removes the invalid block. A successful challenger receives a financial reward, typically derived from the stake of the fraudulent party. Conversely, false challengers lose their stake, discouraging frivolous disputes.

Finality on Arbitrum is delayed by the challenge period, which affects asset bridging. Users must wait until this period ends before safely withdrawing assets to Ethereum. Although this delay is inherent to optimistic security models, it is a practical consideration in bridge comparisons.

\subsection{Wormhole protocol}

\textit{Wormhole} defines itself as a cross-chain message-passing protocol that enables communication between heterogeneous blockchains. Bridging protocols are built on top of \textit{Wormhole} to facilitate token and data transfers across chains.  

On-chain, \textit{Wormhole} uses a set of smart contracts deployed to both the source and destination chains. The emitter contract on the source chain initiates a transfer by invoking the \textit{`publishMessage`} method on the \textit{Wormhole} core contract. These core contracts are responsible for emitting and verifying cross-chain messages, forming the trust set $\tau = \{\text{SCs}\}$ on both ends. The user relies on the correct and secure execution of these smart contracts for the protocol’s integrity.  

The off-chain implementation $\mathcal{T}_\textit{off}$ relies heavily on a guardian network, composed of 19 independent nodes operated by various institutional entities. These guardians observe on-chain activity, verify messages, and collectively sign a Verifiable Action Approval (VAA). A quorum of at least 13 of 19 guardians must sign the VAA for it to be considered valid. This introduces an additional trust assumption, so the full trust set includes off-chain actors, i.e., $\tau \neq \{\}$. In practice, this makes \textit{Wormhole} a trusted bridge, as its security hinges on the honest majority of the guardian set.  

Once the VAA is generated, a relayer (which may be a user or a third party) delivers the signed message to the destination chain. There, the \textit{Wormhole} core contract verifies the guardian signatures and executes the associated instructions via the target smart contract. As with the source chain, the destination chain also has a trust set $\tau = \{\text{SCs}\}$.  

Fee and time considerations in \textit{Wormhole} are variable. Since the relayer role is permissionless, users or third parties may bear the cost of relaying messages across chains, and some applications built on \textit{Wormhole} subsidize or automate these steps. Finality and latency depend on the source chain’s block confirmation time, guardian network processing, and delivery to the destination chain. As such, it does not provide instant finality but is generally faster than optimistic rollups that require challenge periods.  

\textit{Wormhole} connects a broad range of blockchains, including Solana, Algorand, Near \cite{near_protocol}, Aptos \cite{aptos}, Sui \cite{sui}, and CosmWasm chains \cite{cosmwasm}, as well as EVM-compatible blockchains such as Ethereum, Arbitrum, Avalanche, Polygon, Base \cite{base_blockchain}, and Optimism. This makes it one of the most widely integrated cross-chain messaging protocols.

\subsection{Celer \textit{cBridge}}

Celer \textit{cBridge} offers two models for cross-chain asset transfers: a liquidity-pool-based model and an OpenCanonical token bridging model. In the pool-based model, liquidity providers lock assets into pools on both the source and destination chains, allowing users to swap assets between chains directly. In the OpenCanonical model, token transfers follow a mint-and-burn approach. On the source chain, tokens are locked in a smart contract called \textit{TokenVault}, initiating the bridging process. The trust set on the source chain is $\tau = \{\text{SCs}\}$, as users must rely on the correct execution of these smart contracts.

The off-chain implementation $\mathcal{T}_\text{off}$ involves the State Guardian Network (SGN), a Cosmos-based proof-of-stake blockchain built on Tendermint. SGN acts as a validator and orchestrator of the bridge, confirming transactions off-chain. Since SGN is a third blockchain mediating between the source and destination chains, \textit{cBridge} operates as a sidechain-based bridge. The system’s trust assumptions, therefore, include the consensus protocol of SGN, denoted $\Re_{\text{SC}}$, as well as the live correctness of validators ($LC$), message propagation ($MP$), and the honest majority of SGN nodes ($N$). The resulting trust set is $\tau = (\{LC, MP\} \cap \{N\}) \cup \Re_{\text{SC}}$, characterizing \textit{cBridge} as a trust-minimized bridge relative to fully trusted notary-based models.

After validation by SGN, the transaction is relayed to the \textit{PeggedToken} contract on the destination chain, which mints pegged tokens to complete the transfer. The destination chain's trust set is again $\tau = \{\text{SCs}\}$, as users rely on the correct behavior of the deployed smart contracts.

An additional layer of optional security is available through an optimistic-rollup-style model. In this mode, before SGN finalizes a transfer, there is a buffer delay period during which transactions can be independently verified. If inconsistencies are detected, the transfer can be halted, offering an added layer of fraud prevention.

Fee and time considerations for \textit{cBridge} vary by model and configuration. Both \textit{xAsset} (mint-and-burn) and \textit{xLiquidity} (pool-based) bridges impose a base fee and a protocol fee. The base fee includes the destination chain’s gas costs, while the protocol fee, which ranges from 0 percent to 2 percent, depends on the source and destination chains and the transfer amount. This fee compensates SGN validators and stakers for securing the protocol.

\textit{cBridge} currently supports a wide range of chains (41 at the time of this writing), including Ethereum, Arbitrum, Avalanche, BNB Chain, Celo \cite{celo}, Polygon, and Fantom, among others.

\subsection{Avalanche Bridge}

The \textit{Avalanche Bridge} is designed to enable fast, secure, and low-cost asset transfers between Avalanche and external networks such as Ethereum and Bitcoin. On the source chain, a smart contract is invoked to initiate the bridging process, forming a trust set $\tau = \{\text{SCs}\}$, where users rely on the correct execution of the deployed smart contracts.

The off-chain implementation $\mathcal{T}_\text{off}$ of the \textit{Avalanche Bridge} is distinctive and combines a committee-based relayer design with hardware-based verification. A group of eight relayers, known as Wardens, monitor on-chain events and submit signed messages to a secure enclave powered by Intel SGX. This enclave, running verified and attested code, acts as a trusted execution environment that verifies the correctness of submitted data. Once at least six out of eight Wardens agree on the same transaction, the SGX enclave signs off on the transfer. This architecture introduces off-chain trust assumptions, and the trust set satisfies $\tau \ne \{\}$, making \textit{Avalanche Bridge} a trusted bridge. Although the Warden set is geographically and institutionally distributed, the need for a quorum and reliance on Intel SGX introduces centralization risks in the event of collusion or compromise.

On the destination chain, the implementation $\mathcal{T}_\text{dest}$ involves a validator-controlled process. After off-chain validation, a smart contract is called to mint the bridged tokens. While this minting process is smart contract-based, finality is dependent on the decentralized consensus of the underlying blockchain (e.g., Avalanche C-Chain), where validators validate the resulting state. As this process is fully decentralized, the trust set at this stage can be written as $\tau = \{\}$.

Fee and time considerations vary depending on the asset and direction of transfer. Transfers from Ethereum to Avalanche typically take around 15–20 minutes, while Bitcoin transfers may take up to 1 hour due to Bitcoin’s block time. On the Avalanche side, confirmations complete within seconds. ERC-20 tokens sent from Ethereum to Avalanche incur a 0.025 percent fee (minimum USD 3, maximum USD 250), while USDC is charged at a reduced rate of 0.02 percent (minimum USD 3, maximum USD 250). Bridging WBTC or WETH from Ethereum to Avalanche incurs a flat fee of USD 3. Transfers in the opposite direction, from Avalanche to Ethereum, are charged a 0.1 percent fee (minimum USD 12, maximum USD 1000). WBTC and WETH offboarding carries a flat USD 20 fee. Native Bitcoin transfers to Avalanche incur a minimum fee of approximately USD 3 in BTC and result in the minting of BTC.b tokens. Returning BTC.b from Avalanche to the Bitcoin network incurs a fee of approximately USD 20 plus Bitcoin miner fees, with a minimum transfer threshold in place to ensure the transaction is economically viable.

The \textit{Avalanche Bridge} currently supports bi-directional cross-chain transfers between Ethereum or Bitcoin and the Avalanche C-Chain. ERC-20 tokens bridged from Ethereum are wrapped and appear with a \texttt{.e} suffix (e.g., \texttt{USDC.e}, \texttt{WBTC.e}), while Bitcoin is represented as \texttt{BTC.b}. One exception is USDC, which now uses Circle’s Cross-Chain Transfer Protocol (CCTP) and is natively burned and minted on each side, removing the \texttt{.e} suffix. Avalanche-native assets cannot be bridged to other chains using the Avalanche Bridge.

\subsection{Multichain}

\textit{Multichain}, formerly known as Anyswap, presents itself as infrastructure for arbitrary cross-chain interactions. On the source chain, the bridge locks the user's assets in a smart contract that operates under a Secure Multi-Party Computation (SMPC) scheme. This smart contract is referred to by \textit{Multichain} as a Decentralized Management Account, which securely holds assets during the bridging process. The source chain trust set is $\tau = \{\text{SCs}\}$, as users must trust the correct behavior of this contract.

The off-chain implementation $\mathcal{T}_\text{off}$ relies on a network of SMPC nodes, which collectively validate and sign transactions before initiating actions on the destination chain. These nodes serve as notaries in the bridging process, making the trust set $\tau \ne \{\}$, and thus \textit{Multichain} is categorized as a trusted bridge. Because the off-chain signing controls minting and releasing assets, the integrity of this notary network is critical to the bridge’s security.

On the destination chain, once a transaction is validated off-chain by the SMPC network, a smart contract responsible for minting wrapped assets is triggered. The tokens are minted 1:1 relative to the assets locked in the Decentralized Management Account on the source chain. The trust set on the destination chain is again $\tau = \{\text{SCs}\}$, reflecting reliance on the execution of deployed smart contracts. For reverse transfers, the wrapped token contract burns tokens on the destination chain, and the SMPC nodes authorize the release of the original assets from the source chain’s smart contract.

Fee and time considerations depend on the chain and transaction size. For non-Ethereum chains, the cross-chain fee ranges from USD 0.9 to USD 1.9. For Ethereum, a 0.1 percent fee applies, with a minimum of USD 40 and a maximum of USD 1000. Minimum transfer amounts are USD 12 for non-Ethereum destinations and USD 50 when bridging to Ethereum. Maximum transfer size is capped at USD 20 million, though transactions above USD 5 million may incur delays of up to 12 hours.

\textit{Multichain} supports a broad range of blockchains, including Bitcoin, Litecoin \cite{litecoin}, Blocknet \cite{blocknet}, ColossusXT \cite{colossusxt}, Terra \cite{terra_blockchain}, and many EVM-compatible networks such as Ethereum, Avalanche C-Chain, Binance Smart Chain, Celo, Polygon, and Arbitrum. This wide integration makes \textit{Multichain} one of the most expansive cross-chain bridges currently available.

\begin{table}[hbt!]\centering
\caption{Connectivity matrix of major blockchain networks based on existing bridge protocols. Each cell indicates whether a direct bridge exists and, if so, categorizes the bridge by its trust model: \textcolor{green}{TL: Trustless}. Verification relies solely on the blockchains’ own consensus (no third-party custodian), \textcolor{orange!60}{TM: Trust-Minimized}. Decentralized validators or bonded operators secure the bridge (additional assumptions, but no single custodian), and \textcolor{red}{TR: Trusted}. A custodial or centralized entity/committee secures the bridge (users must fully trust this intermediary).}
\label{tab:bridgetrustmatrix}
\renewcommand{\arraystretch}{1.1}
\setlength{\tabcolsep}{4pt}
\newcommand{\TL}{\cellcolor{green!30}TL}
\newcommand{\TM}{\cellcolor{orange!30}TM}
\newcommand{\TR}{\cellcolor{red!30}TR}
\resizebox{\linewidth}{!}{
\begin{tabular}{lccccccccccccccc}
\toprule
Chains & \textbf{BTC} & \textbf{ETH} & \textbf{SOL} & \textbf{AVAX} & \textbf{ATOM} & \textbf{NEAR} & \textbf{DOT} & \textbf{XRP} & \textbf{XLM} & \textbf{ALGO} & \textbf{ADA} & \textbf{ARB} & \textbf{OP} & \textbf{MATIC} & \textbf{BNB} \\ 
\midrule
\textbf{BTC}     & -- & \TM & \TM & \TM & \TM & \TL & \TL & \TM &     & \TR & \TL & \TL & \TM & \TM & \TM \\ 
\textbf{ETH}     & \TM & -- & \TM & \TM & \TM & \TL & \TL & \TM & \TM & \TM & \TM & \TL & \TL & \TR & \TM \\ 
\textbf{SOL}     & \TM & \TM & -- & \TM &      & \TM &      & \TM & \TM & \TM &      & \TM & \TM & \TM & \TM \\ 
\textbf{AVAX}    & \TM & \TM & \TM & -- & \TM & \TM &      & \TM & \TM & \TM & \TM & \TM & \TM & \TM & \TM \\ 
\textbf{ATOM}    & \TM & \TM &      & \TM & -- &      & \TL &      &      &      &      & \TM & \TM & \TM & \TM \\ 
\textbf{NEAR}    & \TL & \TL & \TM & \TM &      & -- &      &      &      &      &      & \TM & \TM & \TM & \TM \\ 
\textbf{DOT}     & \TL & \TL &      &      & \TL &      & -- &      & \TM &      &      & \TM & \TM & \TM & \TM \\ 
\textbf{XRP}     & \TM & \TM & \TM & \TM &      &      &      & -- & \TM &      & \TM & \TM & \TM & \TM & \TM \\ 
\textbf{XLM}     &     & \TM & \TM & \TM &      &      & \TM & \TM & -- &      &      &      &      & \TM & \TM \\ 
\textbf{ALGO}    & \TR & \TM & \TM & \TM &      &      &      &      &      & -- &      & \TM & \TM & \TM & \TM \\ 
\textbf{ADA}     & \TL & \TM &      & \TM &      &      &      & \TM &      &      & -- & \TM & \TM & \TM & \TM \\ 
\textbf{ARB}     & \TL & \TL & \TM & \TM & \TM & \TM & \TM & \TM &      & \TM & \TM & -- & \TM & \TM & \TM \\ 
\textbf{OP}      & \TM & \TL & \TM & \TM & \TM & \TM & \TM & \TM &      & \TM & \TM & \TM & -- & \TM & \TM \\ 
\textbf{MATIC}   & \TM & \TR & \TM & \TM & \TM & \TM & \TM & \TM & \TM & \TM & \TM & \TM & \TM & -- & \TM \\ 
\textbf{BNB}     & \TM & \TM & \TM & \TM & \TM & \TM & \TM & \TM & \TM & \TM & \TM & \TM & \TM & \TM & -- \\ 
\bottomrule
\end{tabular}
}
\end{table}

Table~\ref{tab:bridgetrustmatrix} shows a connectivity matrix for major blockchains and selected Layer-2 networks. Each cell is color-coded to indicate the nature of the bridge connection (if one exists) between the row and column blockchains: Trustless, Trust-Minimized, or Trusted. For example, \textit{Rainbow Bridge} connecting Ethereum and NEAR is trustless, while \textit{Binance Bridge} (BNB Chain) and \textit{Polygon PoS} Bridge are trusted/custodial solutions. \textit{Cosmos IBC} and Polkadot’s upcoming \textit{Snowbridge} are also designed to be trustless. In contrast, multi-signature or externally validated bridges like \textit{Wormhole} and \textit{Multichain} are trust-minimized (they decentralize validation but still rely on additional trust assumptions beyond the base blockchains. Blank cells indicate no widely used direct bridge between those networks. All entries are based on up-to-date bridge protocols (e.g., \textit{Wormhole}, \textit{Multichain}, \textit{LayerZero}, \textit{IBC}, etc.) and public documentation of their trust models.

\section{State-of-Art for Blockchain Bridges}\label{appendix:state} Blockchain interoperability has created an ecosystem of bridges linking most L1 and L2 networks. For example, the Wormhole protocol alone connects 23+ blockchains across six different smart contract runtimes \cite{wormhole2025}, and the Axelar network bridges over 40 chains~\cite{axelar2025}. In parallel, Ethereum serves as a hub in this web; many chains deploy bridges to Ethereum to tap its liquidity and user base. Even previously siloed platforms are integrating. For instance, Cardano \cite{cardano} has explored the Inter-Blockchain Communication (IBC) Protocol~\cite{cosmosIBC2024} to join the Cosmos network.  

\subsubsection{Trust Models}  Classic examples of trustless bridges fall into two categories: L1-to-L1 bridges, such as IBC, and L1-to-L2 bridges, such as in Arbitrum. On IBC channels, light clients on each chain validate each other’s consensus state~\cite{ibc2025}, so IBC’s security reduces to the underlying chains’ security with no additional trusted parties (native chain security). Rollup-based L1-to-L2 bridges fall into optimistic and zero-knowledge-based bridges~\cite{rollup2024}. A zero-knowledge rollup bridge includes a validity proof with every state update, so Ethereum only accepts withdrawals with a valid proof, whereas an optimistic rollup bridge relies on fraud proofs with a challenge period. In optimistic rollup bridges, the bridge inherits the source chain’s security and does not introduce new trust assumptions, aside from at least one honest watcher in optimistic systems \cite{thibault2022blockchain}. Indeed, industry analyses consider Ethereum’s rollup bridges (L1-to-L2) among the most trustless interoperability solutions available~\cite{kiepuszewski2022rollups}. 

Rollups cannot facilitate bridging between L1 chains with different runtimes (e.g., Bitcoin and Ethereum). As a result, most L1-to-L1 connectivity relies on trust-minimized intermediaries (i.e., networks of validators, relays, or oracles) that verify transactions across chains. Users must trust that a majority of these actors behave honestly, often incentivized through mechanisms like staked collateral or slashing penalties. Examples include \textit{Axelar}~\cite{axelar}, \textit{Wormhole}~\cite{wormhole}, \textit{LayerZero}~\cite{layerzero}, \textit{Multichain}~\cite{multichain}, and Celer’s \textit{cBridge}~\cite{celer_cbridge}. These permissioned or federated bridges require a predefined set of entities to authorize transfers and offer stronger assurances than single custodians, but fall short of the security provided by native chain verification.  For instance, bridging assets from Ethereum to \textit{BNB Chain} \cite{bnb_chain} or Solana often involves a validator coalition verifying the source chain deposit and minting a wrapped asset on the destination chain. 

At the far end of the spectrum are trusted bridges, which include custodial bridges (e.g., wrapped Bitcoin via \textit{BitGo}’s WBTC, where a single custodian holds the asset reserve) and some early L1-to-L1 bridges run by a single team or exchange. A notable example was Binance’s original bridge connecting Ethereum and \textit{BSC} (\textit{BNB Smart Chain})~\cite{binancebridge2023}, essentially operated by Binance as a trusted custodian. Similarly, the \textit{Ronin bridge} for Axie Infinity (prior to its 2022 hack) relied on just nine validators controlled by a single organization (Sky Mavis) and its partners, effectively a federated multisig under one entity’s control. Such trusted models have minimal decentralization: users must trust that the custodian or signers will never misbehave, as there are no protocol-level penalties or verifications of their actions. While simple and fast to deploy, these bridges have the weakest security guarantees and have often proven fragile under attack, such as the \textit{Ronin bridge} hack of 2022 (see Table~\ref{tab:bridge_attacks}). Even when the underlying bridge technology is secure, adversaries may compromise physical infrastructure to exfiltrate private keys.


 \begin{table}[h]
\centering
\caption{Documented security incidents and failures in cross-chain bridge and routing architectures. Most attacks involve violations of the cross-chain causality prior.}
\label{tab:bridge_attacks}
\resizebox{\linewidth}{!}{
\begin{tabular}{llcrllll}
\toprule
 Bridge & Architecture Type & Date ($\downarrow$) & Loss (USD) & Technique & Violation &Target &Vector \\
\midrule
 \href{https://rekt.news/thorchain-rekt2} {Thorchain} & Sidechain/Relay & 2021-06-29 & 140K & False top-up&Causality& Source chain &V3 \\
 \href{https://rekt.news/chainswap-rekt} {ChainSwap} & Notary & 2021-07-02 & 800K & Contract vulnerability &Consistency& Source chain &V3 \\
 \href{https://rekt.news/chainswap-rekt} {ChainSwap 2} & Notary & 2021-07-11 & 4M & Contract vulnerability &Consistency& Source chain &V3\\
 \href{https://rekt.news/multichain-rekt2} {Multichain} & MPC (custodial) & 2021-07-10 & 7.9M & Private key compromised (Bad ECDSA) &Causality& Source chain &V13\\
 \href{https://rekt.news/thorchain-rekt2} {Thorchain 2} & Sidechain/Relay & 2021-07-15 & 7.6M & False top-up &Causality& Source chain &V9\\
 \href{https://rekt.news/thorchain-rekt2} {Thorchain 3} & Sidechain/Relay & 2021-07-22 & 8M & Refund logic exploit&Causality& Source chain &V3\\
 \href{https://rekt.news/thorchain-rekt2} {Thorchain 4} & Sidechain/Relay & 2021-07-24 & 76K & Phishing attack &Causality& Off-chain relayer &V23\\
 \href{https://rekt.news/poly-network-rekt2} {Poly Network} & Notary & 2021-08-10 & 611M & Trusted state root exploit&Causality& Source chain &V9\\
 \href{https://rekt.news/poly-network-rekt2} {pNetwork} & Notary & 2021-09-20 & 13M & Inconsistent event parsing &Consistency& Off-chain relayer &V6\\
 Plasma Bridge & Sidechain/Relay & 2021-10-05 & None & Reused burn proof&Consistency& Destination chain&V15\\
 Optics Bridge & Notary & 2021-11-23 & None & Multi-signature permission vulnerability&Causality& Off-chain relayer &V3\\
 \href{https://rekt.news/multichain-rekt2} {Multichain 2} & MPC (custodial) & 2022-01-18 & 1.4M & Token contract vulnerability &Consistency& Source chain &V3\\
 \href{https://rekt.news/qubit-rekt} {Qubit Finance} & Relay & 2022-01-28 & 80M & Deposit function exploit&Causality& Source chain &V3\\
 \href{https://rekt.news/wormhole-rekt} {Wormhole} & Validator-based & 2022-02-02 & 326M & Signature exploit&Causality& Destination chain &V12\\
 \href{https://rekt.news/meter-rekt} {Meter} & Sidechain/Relay & 2022-02-06 & 7.7M & Deposit function exploit&Causality& Source chain &V3\\
 \href{https://rekt.news/ronin-rekt} {Ronin} & Sidechain & 2022-03-23 & 624M & Private key compromised (social engineering) &Causality& Off-chain relayer&V13\\
 Marvin Inu & Custodial & 2022-04-11 & 350K & Private key compromise &Causality& Source chain &V13\\
 QANX Bridge & Notary & 2022-05-18 & 2.2M & Deploy fake bridge contract onto BSC&Consistency& Source chain&V3\\
 \href{https://rekt.news/harmony-rekt} {Horizon Bridge} & Notary & 2022-06-23 & 100M & Private key compromised (unknown method)&Causality& Off-chain relayer&V12\\
 \href{https://rekt.news/nomad-rekt}{Nomad} & Relay & 2022-08-01 & 190M & Trusted state root exploit&Causality& Destination chain &V10\\
 \href{https://cointelegraph.com/news/celer-network-shuts-down-bridge-over-potential-dns-hijacking}{Celer} & Sidechain/Relay & 2022-08-18 & 240K & DNS cache poisoning attack on frontend UI appprox &Causality& Off-chain relayer &V23\\
 \href{https://chainbulletin.com/ethw-replay-exploit-caused-by-omni-contract-vulnerability}{Omni Bridge} & Notary & 2022-09-18 & 290K & Possible ChainID vulnerability&Consistency& Destination chain&V15\\
 \href{https://rekt.news/bnb-bridge-rekt}{Binance Bridge} & Custodial & 2022-10-06 & 570M & Proof Verifier Bug&Consistency& Off-chain relayer &V9 \\
 \href{https://finance.yahoo.com/news/quantum-resistant-blockchain-qanplatform-suffers-121549717.html}{QANX Bridge 2} & Notary & 2022-10-11 & 2M & Weak address key vulnerability&Consistency& Destination chain &V13\\
 \href{https://rekt.news/thorchain-rekt2} {Thorchain 5} & Sidechain/Relay & 2022-10-28 & None & Network interruption &Causality& Off-chain relayer &V17\\
 \href{https://rekt.news/poly-network-rekt2} {pNetwork} & Hash-locking & 2022-11-04 & 10.8M & Misconfiguration of the pNetwork-powered bridge &Consistency& Token contract &V3\\
 \href{https://rekt.news/multichain-rekt2} {Multichain 3} & MPC (custodial) & 2023-02-15 & 130K & Rush attack &Consistency& Source chain &V13\\
 \href{https://rekt.news/dexible-rekt}{Dexible} & Trade router& 2023-02-20 & 2M & Unchecked destination address &Consistency& Source chain &V5\\
 \href{https://blog.solidityscan.com/allbridge-hack-analysis-improper-business-logic-564fbadf38b2}{Allbridge} & Sidechain/Relay & 2023-04-02 & 570K & Logic flaw in withdraw function &Consistency& Destination chain&V2 \\
  \href{https://cointelegraph.com/news/jump-crypto-finds-double-voting-vulnerability-in-celer-s-sgn}{Celer} & Hybrid& 2023-05-24 & None & Double-voting vulnerability & Causality & Off-chain relayer &V18\\

 \href{https://rekt.news/poly-network-rekt2} {Poly Network} & Notary & 2023-07-02 & 10M & Private key compromised &Causality& Off-chain relayer &V13\\
 \href{https://rekt.news/poly-network-rekt2}{Poly Network} & Notary & 2023-07-03 & 4.4M & Compromised multisig &Causality& Off-chain component &V13\\
 \href{https://rekt.news/multichain-rekt2}{Multichain 4} & MPC (custodial) & 2023-07-07 & 126M & Compromised address &Causality& Off-chain component &V13\\
 \href{https://rekt.news/multichain-r3kt}{Multichain 5} & MPC (custodial) & 2023-07-14 & unknown & State seized multisigs &Causality& Off-chain component &V12\\
 \href{https://rekt.news/shibarium-bridge-rekt}{Shibarium} & L2 Sequencer Bridge & 2023-08-17 & None & Frozen withdrawals due to L2 bug &Causality& Off-chain component &V17\\
 \href{https://rekt.news/hecoi-htx-rekt}{HTX} & MPC (custodial) & 2023-11-22 & 12.5M & Hot wallet compromised &Causality& Off-chain component &V13\\
  \href{https://rekt.news/heco-htx-rekt}{HECO Bridge} & MPC (custodial) & 2023-11-22 & 86.6M & Hot wallet compromised &Causality& Off-chain component &V13\\
\href{https://rekt.news/hypr-network-rekt}{HYPR Network} & Sidechain/Relay & 2023-12-13 & 220K & Vulnerability in external code dependency &Consistency& Destination chain&V7 \\
 \href{https://rekt.news/socket-rekt}{Socket's Bungee} & Aggregator (router) & 2024-01-16 & 3.3M & Compromised multisig &Consistency& Source chain &V12\\
\href{https://cointelegraph.com/news/alex-bridge-bnb-drained-after-suspicious-upgrade-certik}{Alex} &Notary&2024-05-14&4.3M&Compromised deployer&Causality&Source chain&V12, V7\\
 \href{https://rekt.news/lifi-jumper-rekt}{Lifi/Jumper} & Aggregator (meta-router) & 2024-07-16 & 9.73M & Buggy function &Consistency& Source chain &V2\\
  \href{https://rekt.news/roninnetwork-rektII}{Ronin 2} & Sidechain & 2024-08-06 & 12M & Whitehat MEV attack (parameter error in update) &Causality& Source chain &V4\\
\href{https://protos.com/feg-token-holders-in-despair-after-third-hack-causes-99-dump/}{Feed Every Gorilla} & Notary (Wormhole-based) & 2024-12-30 & 1.3M & Message spoofing & Causality & Source chain &V9\\
\bottomrule
\end{tabular}
}
\end{table}

\subsubsection{Connectivity Trends} Appendix \autoref{tab:bridgetrustmatrix} shows the current bridge topology with clear patterns. Ecosystems with homogeneous technology (e.g., shared virtual machine) use trustless native bridges among themselves, whereas heterogeneous connections (e.g., Bitcoin-Ethereum) gravitate toward trust-minimized hubs. The Cosmos ecosystem is a prime example of the former: dozens of Cosmos-SDK chains (e.g., Osmosis \cite{Osmosis}, Cosmos Hub, Cronos \cite{cronos}, and Juno \cite{cronos}) all interconnect via IBC channels with no external mediators.

To bridge out to Ethereum and other ecosystems, Cosmos projects deployed separate bridge modules and networks. For example, the \textit{Gravity Bridge chain} \cite{gravity_bridge} and Axelar network act as Cosmos-Ethereum L1-to-L1 bridges by having their validator sets observe events on Ethereum and vice versa~\cite{gravity2023}. These are trust-minimized solutions (essentially multisig validator bridges) added onto Cosmos to import assets like USDC, DAI, and WETH into the IBC realm.  

Ethereum and its orbiting chains illustrate a different connectivity paradigm. Virtually every major L1 (Solana, BNB Chain, Avalanche, Polygon, Fantom \cite{fantom}, etc.) has at least one bridge to Ethereum, given the high value of assets and liquidity on Ethereum.  For example, \textit{Wormhole} operates as a cross-chain message network connecting Ethereum with Solana, BSC, Avalanche, Polygon, and others, using a guardian consensus (a set of 19 nodes) to validate transfers. When a user moves ETH from Ethereum to Solana via the L1-to-L1 \textit{Wormhole}, the ETH is locked in a \textit{Wormhole} contract on Ethereum, and a wrapped ETH is minted on Solana, based on a signed attestation by the guardian network. This design is trust-minimized (no single custodian, but users trust the majority of guardians). 

Other popular Ethereum bridges follow similar patterns: \textit{Multichain} (Anyswap) uses a rotated set of Multi Party Computation signers to custody and mint assets between Ethereum and EVM chains; Celer \textit{cBridge} uses a Proof-of-Stake validator set to oversee transfers. In all these cases, Ethereum-to-L1 bridges tend to introduce a separate middleware layer.

Outside the Ethereum/Cosmos spheres, other ecosystems have pursued their own bridging approaches with mixed strategies. Solana, for instance, developed ties to Ethereum and others mainly through third-party bridge networks. \textit{Wormhole}’s origin was as the Solana-Ethereum L1-to-L1 bridge (also known as \textit{Portal}), and it remains the primary connector between Solana and multiple other chains. 

One constraint Solana faced is that its very high throughput (2.9K - 65K tx/sec) and non-EVM-based execution environment meant fewer independent bridge options. The ecosystem largely relied on \textit{Wormhole}, which introduced a single point of failure for connectivity. Recent projects are working on light-client bridges (e.g., using Solana's clients), and even Cosmos’s IBC was adapted in 2023 to work from Solana via an adapter chain. However, Solana’s cross-chain connectivity is primarily driven by multi-signature bridge protocols rather than native, trustless channels. The BNB Chain ecosystem likewise began with a highly centralized bridging model. BNB Chain comprises a dual-chain system, consisting of an application-focused Beacon Chain and an EVM-compatible Smart Chain, bridged by the \textquote{Token Hub}. The BSC Token Hub initially operated under the tight control of a few validators, as evidenced by the fact that a successful attack in October 6, 2022 was able to forge just two messages and mint 2 million BNB (see \autoref{tab:bridge_attacks}, an exploit that even forced a temporary halt to the chain.  In response, BNB Chain has moved to increase the security of its bridges (e.g., raising multisig thresholds and improving key management) and has also encouraged the use of external bridges. 

Nowadays, BNB Chain is connected to Ethereum and others via both Binance’s official L1-to-L1 bridge (still centrally governed) and alternatives like \textit{Celer}, \textit{deBridge}~\cite{debridge}, and \textit{LayerZero}~\cite{layerzero}, which are more decentralized. Nevertheless, BNB’s connectivity strategy remains cautious: its native bridge prioritizes speed and user experience (at the cost of trust), whereas third-party solutions offer decentralization but introduce dependency on external validator sets. We observe a similar pattern in other ecosystems, such as Polygon (the \textit{Polygon PoS} bridge utilizes a 5-of-8 multisig to checkpoint withdrawals, a somewhat trusted design, alongside newer trust-minimized options) and Tron (which relies on custodial bridges for assets like BTC/ETH via BitGo, or exchange-mediated transfers). In contrast, Polkadot \cite{polkadot} takes a route closer to Cosmos: all parachains in a Polkadot parachain cluster are natively interoperable via the XCM (Cross-Consensus Message) format~\cite{xcm}, which is \textquote{trustless and secure} for in-network messaging. Polkadot’s relay chain ensures the validity of messages between parachains, so within this ecosystem, cross-chain actions (such as token transfers and remote contract calls) don’t require external trust. However, bridging Polkadot to external chains still entails separate bridge projects (often with their own validator sets or light clients). Efforts are underway (e.g., Snowfork and Darwinia for Ethereum-Polkadot L1-to-L1 bridges) to make those as trust-minimized as possible, potentially even integrating light-client verification.

\section{Descriptions of \textdagger Attack Vectors}
\label{appsec:vectors}
In Table~\ref{tab:attack_vectors_categorized}, we used \textdagger~to indicate that these vulnerabilities only apply under certain design conditions. We now elaborate on those conditions.

\noindent\textbf{V1: Fake Burn/Lock Proofs.}
Tokens are locked or burned on the source chain, X, before minting can occur on the destination chain, Y. A cryptographic proof sent from X to Y is required for confirmation before minting can occur. This attack involves forging the proof somehow to trick the Y chain into minting before any proper burn or lock has actually happened.

\noindent\textbf{V2: Malicious Transaction Modification.}
Transaction details (addresses, number of tokens, gas fees, etc.) are sent from the source chain to the destination chain during a transaction. Malicious transaction modification involves somehow intercepting this detailed payload and asserting some fake payload. This can be done off-chain with protocols that include some centralized sequencer, or on-chain when the destination bridge is looking for the detailed payload.

\noindent\textbf{V3: Reentry Attack}
Occurs when some external function is called multiple times before the initial execution has completed. This attack vector was used in the famous DAO attack, where the token minting was executed many times before the balance was checked and updated.

\noindent\textbf{V4: Integer Overflow/Underflow}
When overflow or underflow remain unchecked, attackers can induce attacks by creating unexpected results. A wrap-around can also happen, where the maximum value can wrap around to the minimum value, or vice versa, creating room for exploitation.

\noindent\textbf{V5: Access Control Flaws}
 Access control is very important in maintaining control over a contract, and if unchecked, an attacker can impersonate a contract owner and manipulate the entire contract.

 \noindent\textbf{V6: Timestamp Manipulation}
 A large number of blockchain protocols require some timestamp with loose limits to allow for network delays. These loose limits allow attackers to falsify timestamps that will still be considered valid by the protocol, leading to vulnerable behavior in the protocol.

 \noindent\textbf{V7: Inconsistent Cross-Chain Transfers}
 Bridges can implement different methods for withdrawing and depositing, such as burning when performing a withdrawal but minting on a deposit. When these operations are different, attackers can take advantage of a lack of one-to-one consistency.

 \noindent\textbf{V8: Replay Attacks}
 Attackers can take a message in transit that was validated, then delay and retransmit the original message in order to assume validation but cause effects unintended by the original message.

 \noindent\textbf{V9: Oracle Manipulation}
 Oracles being used to provide data, such as token pricing, can be points of vulnerability if an attacker were to take control of the oracle and provide fake data, leading to incorrect fees or token amounts during minting or burning. Oracles can provide more data than just token prices, leading to more points of possible attack.

 \noindent\textbf{V10: Consensus Failure}
 Commonly referred to as a 51\% attack, if an attacker were to control 51\% of the mining or staking power (depending on the protocol) then that attacker would be able to create a new chain segment that builds faster than the legitimate chain. This control can lead to false transactions that could mint tokens on other chains.

\noindent\textbf{V11: Malicious Custodian Manipulation}
Custodial wallets are commonly used to hold locked tokens or private keys via a centralized source. Compromising or maliciously colluding with this custodial wallet leads to access of locked tokens, private keys, or other sensitive data.

\noindent\textbf{V12: Key Leakage / Private Key Theft}
Locked assets can exist within bridges accessible by private keys and bridges themselves can hold keys along with their controlling proxy contracts. Any compromising of these keys through phishing or other targeted attacks can create vulnerabilities, potentially giving an attacker full control over a bridge.

\noindent\textbf{V13: Race Condition Attacks}
A vulnerability is created when two or more processes access the same data, such as when withdrawals and deposits are processed without proper synchronization, and conflicting transactions can occur.

\noindent\textbf{V14: Unsafe External Call Exploits}
When a smart contract accesses some external call to third-party contracts or off-chain contracts, these new calls create vulnerability in their own code security that can be exploited.

\noindent\textbf{V15: Forged Account Attacks}
Occurs when an account or address is deliberately created to impersonate the identity of another party. This can lead to the redirection of tokens in a bridge structure if a fake account is validated.

\noindent\textbf{V16: Malicious Event Log Manipulation}
Key actions are logged by smart contracts, typically off-chain. Falsification of the event log can lead to relaying of some fake events that have not happened. This results in tricking the chain into a false state.

\noindent\textbf{V17: Denial of Service Attacks}
Bridge contracts can be flooded with transactions and requests sent by an attacker, overwhelming the bridge protocol and slowing processes for all transactions on the bridge. This could potentially lead to the complete halting of a bridge protocol.

\noindent\textbf{V18: Arithmetic Accuracy Deviation}
Exact calculations are required to maintain consistency across chains. Tiny errors in integer overflow or underflow can cause precision errors, leading to inconsistency between chain states.

\begin{table*}[ht]
    \centering
    \caption{Cross-Chain Bridge Implementations and Architecture}
    \label{tab:bridge_implementations}
    \renewcommand{\arraystretch}{1.2}
    \resizebox{\textwidth}{!}{
    \begin{tabular}{l l l l}
        \toprule
        \textbf{Bridge} & \textbf{Source Chain} (Asset Custody \& Locking) & \textbf{Off-Chain Mechanism} (Cross-Chain Validation) & \textbf{Destination Chain} (Asset Release \& Minting) \\
        \midrule
        \textbf{Allbridge Classic} & Smart Contracts – Assets locked in source-chain bridge contract\footnote{\url{https://li.fi/knowledge-hub/allbridge-a-deep-dive/}} 
        & Notaries/Relays – Small validator set observes lock and signs confirmations\footnote{\url{https://li.fi/knowledge-hub/allbridge-a-deep-dive}} 
        & Smart Contracts – Target chain bridge contract verifies signature and mints tokens\footnote{\url{https://docs.allbridge.io/allbridge-overview}} \\
        
        \textbf{deBridge} & Smart Contracts – Gateway contract escrows assets\footnote{\url{https://docs.debridge.finance/the-debridge-messaging-protocol/protocol-overview}} & Notaries/Relays – Decentralized validator committee signs off-chain\footnote{\url{https://docs.debridge.finance/the-debridge-messaging-protocol/protocol-overview}} & Smart Contracts – Destination contract verifies signatures and releases assets\footnote{\url{https://docs.debridge.finance/the-debridge-messaging-protocol/protocol-overview}} \\
        
        \textbf{Multichain} & Smart Contracts (MPC custody) – Deposited assets held by MPC-controlled vault\footnote{\url{https://docs.multichain.org/}} 
        & Notaries/Relays – Multi-party computation (MPC) validators co-sign transactions\footnote{\url{https://blog.bcas.io/implications-of-cross-chain-bridges-under-mica}}
        & Smart Contracts – Destination contract mints wrapped assets upon MPC authorization\footnote{\url{https://li.fi/knowledge-hub/multichain-a-deep-dive}} \\
        \textbf{Wormhole Bridge} & Smart Contracts – Core contract locks assets or emits transfer message\footnote{\url{https://wormhole.com/docs/learn/infrastructure/core-contracts/}} & Notaries/Relays – Guardian network signs verified action approvals (VAA)\footnote{\url{https://wormhole.com/docs/learn/infrastructure/core-contracts/}} & Smart Contracts – Destination chain contract verifies VAA and mints assets\footnote{\url{https://wormhole.com/docs/learn/infrastructure/vaas/}} \\
        
        \textbf{Avalanche Bridge} & Validator Control – Ethereum escrow EOA with multi-sig custodian\footnote{\url{https://l2beat.com/bridges/projects/avalanche}} & Notaries/Relays – Intel SGX enclave validators approve transactions\footnote{\url{https://l2beat.com/bridges/projects/avalanche}} & Smart Contracts – Destination bridge contract mints assets upon SGX approval\footnote{\url{https://blog.li.fi/avalanche-bridge-a-deep-dive}} \\
        
        \textbf{Ronin Bridge} & Smart Contracts – Ethereum bridge contract locks assets\footnote{\url{https://www.forta.org/blog/ronin-hack}} & Sidechain – Ronin validators (5 of 9 required) approve transfers\footnote{\url{https://www.forta.org/blog/ronin-hack}} & Validator Control – Ronin validators execute mint/burn on sidechain\footnote{\url{https://www.forta.org/blog/ronin-hack}} \\
        
        \textbf{Wanchain Bridge} & Validator Control – Storeman group (25 nodes) holds escrow\footnote{\url{https://docs.wanchain.org/crosschain/security-mechanism}} & Notaries/Relays – Storeman nodes verify transactions using MPC\footnote{\url{https://docs.wanchain.org/wanchain-bridge-nodes/fact-sheet}} & Validator Control – Storeman group authorizes mint/release on target chain\footnote{\url{https://docs.wanchain.org/crosschain/security-mechanism}} \\
        
        \textbf{Connext} & Smart Contracts – NXTP protocol escrow locks tokens\footnote{\url{https://li.fi/knowledge-hub/connext-network-a-deep-dive}} & Relays – Liquidity routers observe lock and relay signed messages\footnote{\url{https://li.fi/knowledge-hub/connext-network-a-deep-dive}} & Smart Contracts – Destination contract releases assets using router liquidity\footnote{\url{https://li.fi/knowledge-hub/connext-network-a-deep-dive}} \\
        
        \textbf{Axelar} & Smart Contracts – Gateway contract on source chain locks assets\footnote{\url{https://www.axelar.network/blog/a-technical-introduction-to-the-axelar-network}} & Sidechain – Axelar validators reach consensus and use threshold cryptography\footnote{\url{https://www.axelar.network/blog/a-technical-introduction-to-the-axelar-network}} & Smart Contracts – Destination gateway releases tokens upon validator approval\footnote{\url{https://www.axelar.network/blog/a-technical-introduction-to-the-axelar-network}} \\
         
        \textbf{Binance Bridge} & Validator Control – Binance custodial wallet holds assets\footnote{\url{https://academy.binance.com/en/articles/what-s-a-blockchain-bridge}} & Notaries/Relays – Binance-controlled relay system detects deposits\footnote{\url{https://academy.binance.com/en/articles/what-s-a-blockchain-bridge}} & Validator Control – Binance mints BTokens pegged to locked assets\footnote{\url{https://academy.binance.com/en/articles/what-s-a-blockchain-bridge}} \\
         
        \textbf{Horizon Bridge} & Validator Control – 2-of-5 multi-signature custody on Ethereum\footnote{\url{https://blockchaingroup.io/guides/harmonys-horizon-bridge-exploit-a-crypto-money-laundering-case-study/}} & Notaries/Relays – Federated validators approve transfers\footnote{\url{https://www.elliptic.co/blog/analysis/over-1-billion-stolen-from-bridges-so-far-in-2022}}
        & Validator Control – Multi-sig bridge mints tokens on Harmony\footnote{\url{https://bridge.harmony.one/faq}} \\
        
        \textbf{BTC Relay} & Validator Control – Bitcoin transactions are verified by BTC Relay’s on-chain light client, but funds must be sent to an externally controlled address\footnote{\url{https://btc-relay.readthedocs.io/en/latest/frequently-asked-questions.html}} & Light Client – BTC Relay runs a Merkle proof verification contract on Ethereum that allows Bitcoin transactions to be verified on-chain\footnote{\url{https://btc-relay.readthedocs.io/en/latest/frequently-asked-questions.html\#what-s-the-simplest-explanation-of-btc-relay}} & Smart Contracts – Ethereum contracts verify Bitcoin transactions and trigger asset issuance or validation\footnote{\url{https://btc-relay.readthedocs.io/en/latest/frequently-asked-questions.html\#what-s-the-simplest-explanation-of-btc-relay}} \\

        \textbf{zkBridge} & Smart Contracts – Source chain tokens are locked in a bridge contract before proof generation\footnote{\url{https://rdi.berkeley.edu/zkp/zkBridge/zkBridge.html}} & Light Client – zk-SNARK proofs verify block headers across chains without external validators\footnote{\url{https://rdi.berkeley.edu/zkp/zkBridge/zkBridge.html}} & Smart Contracts – Destination chain contracts validate proofs and mint or release tokens\footnote{\url{https://rdi.berkeley.edu/zkp/zkBridge/zkBridge.html}} \\

        \textbf{PeaceRelay} & Smart Contracts – Ethereum-based contracts escrow tokens and validate state transitions\footnote{\url{https://blog.kyberswap.com/peacerelay-connecting-the-many-ethereum-blockchains/}} & Light Client – Smart contracts maintain Merkle-Patricia proofs for cross-chain validation\footnote{\url{https://blog.kyberswap.com/peacerelay-connecting-the-many-ethereum-blockchains/}} & Smart Contracts – Contracts on Ethereum Classic or other chains verify proofs and mint assets\footnote{\url{https://blog.kyberswap.com/peacerelay-connecting-the-many-ethereum-blockchains/}} \\

        \textbf{Rainbow Bridge} & Smart Contracts – NEAR and Ethereum smart contracts escrow tokens and initiate cross-chain transfers\footnote{\url{https://dev.near.org/bridge}} & Light Client – Each chain hosts a light client verifying the other chain’s consensus, eliminating external trust\footnote{\url{https://dev.near.org/bridge}} & Smart Contracts – Contracts on NEAR or Ethereum validate and execute transactions upon light-client verification\footnote{\url{https://dev.near.org/bridge}} \\

        \textbf{Cosmos IBC} & Smart Contracts – The IBC module in Cosmos SDK escrows assets and generates proofs\footnote{\url{https://tutorials.cosmos.network/tutorials/6-ibc-dev/}} & Light Client – Chains run light clients of each other, verifying block headers using the Tendermint consensus\footnote{\url{https://tutorials.cosmos.network/tutorials/6-ibc-dev/}} & Smart Contracts – The receiving chain validates IBC messages and mints corresponding assets\footnote{\url{https://tutorials.cosmos.network/tutorials/6-ibc-dev/}} \\

        \textbf{Gravity Bridge} & Validator Control – Ethereum assets are locked in a Gravity.sol contract controlled by the Gravity Bridge validators\footnote{\url{https://www.gravitybridge.net/home-v2}} & Sidechain – Gravity Bridge validators observe Ethereum transactions and reach consensus on transfers\footnote{\url{https://blog.cosmos.network/gravity-is-an-essential-force-of-the-cosmos-aligning-all-planets-in-orbits-in-the-composable-b1ca17de18cc}} & Validator Control – Tokens are minted on Cosmos chains upon validator consensus\footnote{\url{https://blog.cosmos.network/gravity-is-an-essential-force-of-the-cosmos-aligning-all-planets-in-orbits-in-the-composable-b1ca17de18cc}} \\

        \textbf{zkRelay} & Validator Control – Transactions from PoW chains like Bitcoin are observed, with no native smart contracts on the source chain\footnote{\url{https://eprint.iacr.org/2020/433.pdf}} & Light Client – zk-SNARKs are used to verify large batches of Bitcoin headers efficiently\footnote{\url{https://eprint.iacr.org/2020/433.pdf}} & Smart Contracts – Ethereum or other chains verify the zk-proof and execute transactions accordingly\footnote{\url{https://eprint.iacr.org/2020/433.pdf}} \\

        \textbf{Cactus} & Validator Control – A permissioned set of validators controls funds and relays messages\footnote{\url{https://www.hyperledger.org/use/cactus}} & Notaries/Relays – Cross-chain transactions are notarized by trusted relayers in enterprise settings\footnote{\url{https://www.hyperledger.org/use/cactus}} & Validator Control – The target blockchain executes transactions based on validator approvals\footnote{\url{https://www.hyperledger.org/use/cactus}} \\

        \textbf{Celer cBridge} & Smart Contracts – Liquidity pools and lock-mint contracts facilitate asset transfers\footnote{\url{https://cbridge-docs.celer.network/}} & Sidechain/Relayer Network – Celer’s State Guardian Network (SGN) validates transfers via PoS consensus\footnote{\url{https://cbridge-docs.celer.network/}} & Smart Contracts – cBridge smart contracts process releases based on SGN validators’ signed messages\footnote{\url{https://cbridge-docs.celer.network/}} \\

        \textbf{Orbit Bridge} & Validator Control – A federation of Orbit Chain validators governs asset custody\footnote{\url{https://bridge-docs.orbitchain.io/}} & Sidechain – Orbit Chain runs a dedicated blockchain that validates cross-chain transactions\footnote{\url{https://bridge-docs.orbitchain.io/}} & Validator Control – Assets are minted or unlocked on the target chain by Orbit validators\footnote{\url{https://bridge-docs.orbitchain.io/}} \\

        \bottomrule
        
        \end{tabular}
    }

\end{table*}

\noindent\textbf{Function Visibility and Modifier Use.}  

\section{Descriptions of Hacks and Attacks on Bridges}
\label{sec:realAttacks}
In one of the earliest and largest bridge hacks, attackers exploited a flaw in \textit{Poly Network}’s cross-chain transaction handling contract, \textquote{literally tricking the project to hack itself}. The \textit{Poly Network bridge} connected multiple chains, including Ethereum, BSC, Polygon, and others. Its core contract (\textit{EthCrossChainManage}) had an insecure design: it could make an external call to an arbitrary target contract with supplied data, without proper authorization checks. The attacker crafted a call payload that misled the manager contract into invoking \textit{Poly}’s own storage contract (\textit{EthCrossChainData}) and updating the keeper of funds as if the attacker were the owner. As a result, approximately \$611 million worth of cryptocurrency was drained from \textit{Poly Network} in a single attack, with the hacker subsequently returning the funds. The attack flow is illustrated in Figure~\ref{fig:poly-attack-overview}.

\begin{figure}[h]
    \centering
\includegraphics[width=0.9\linewidth]{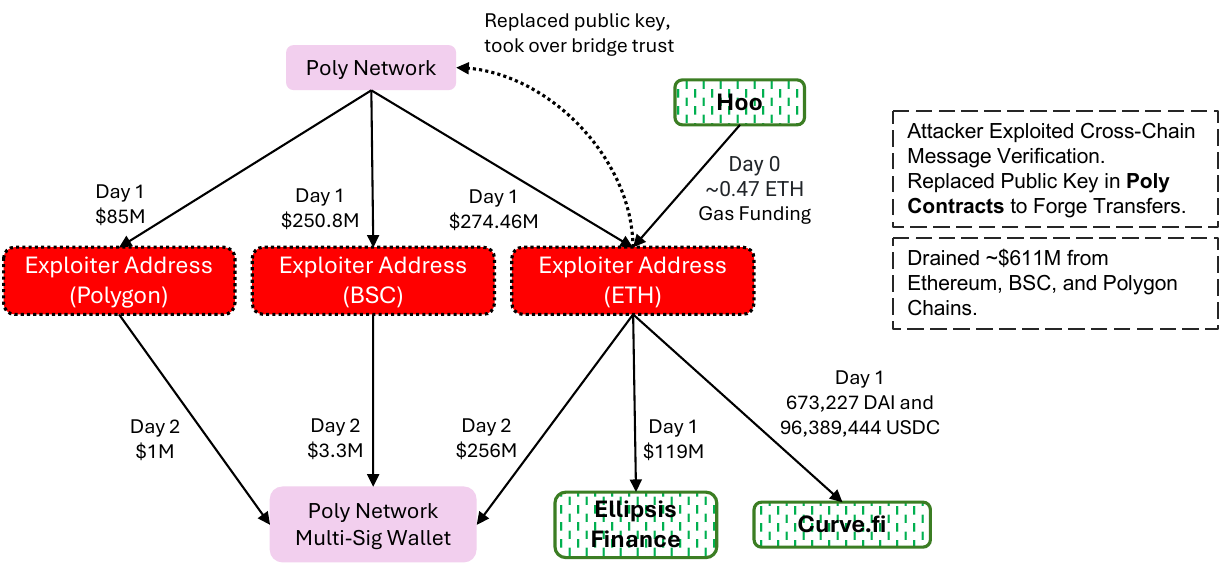}
    \caption{
    Attack flow of the Poly Network exploit.
    }
    \label{fig:poly-attack-overview}
\end{figure}

The \textit{Ronin bridge} to Ethereum was victim to a validator key compromise in an attack now attributed to North Korea’s Lazarus Group. \textit{Ronin} used a 5-of-9 multisig validation for bridge withdrawals. Attackers compromised five private keys through social engineering. Once in control, they issued two fraudulent transactions, draining 173,600 ETH and 25.5M USDC (worth approximately \$624 million at the time) from the \textit{Ronin bridge} in a single stroke. This event, the largest DeFi hack ever, was essentially a failure of the bridge’s trust model: the off-chain validators were assumed honest, but the minimal quorum and centralized key management (a single entity controlled 4 of 9 validators) made it easy for an attacker to breach. The breach went unnoticed for six days until a user discovered it. Figure~\ref{fig:ronin-attack-overview} summarizes the transaction flow of this exploit.

\begin{figure}[h]
    \centering
\includegraphics[width=0.9\linewidth]{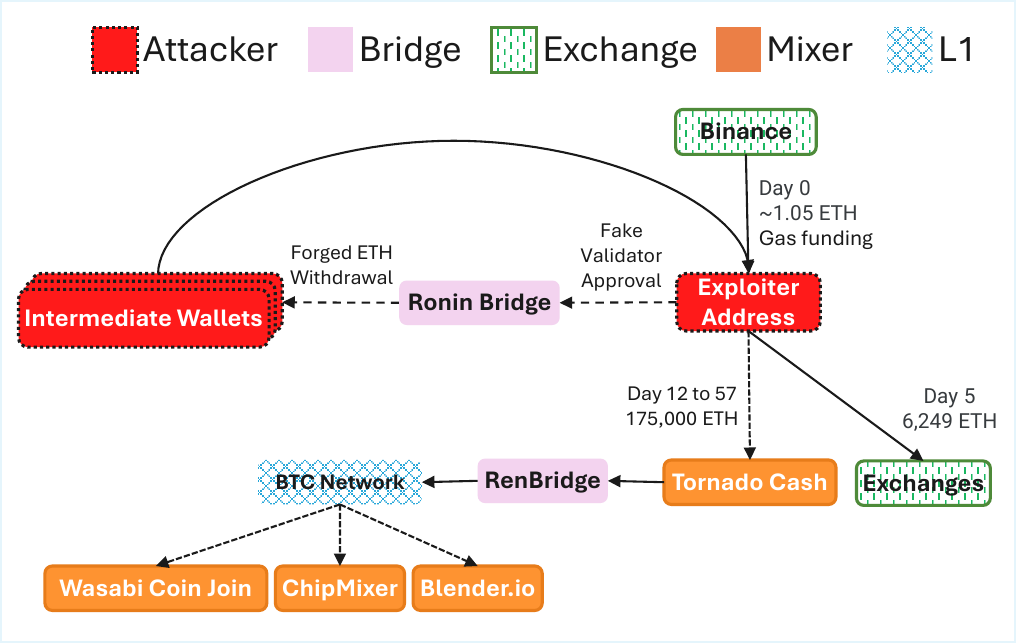}
    \caption{
    Attack flow of the \textit{Ronin} bridge exploit. 
    }
    \label{fig:ronin-attack-overview}
\end{figure}

\textit{Wormhole}’s Ethereum–Solana bridge was hacked for 120,000 WETH (worth \$326 million) in a smart contract bug on the Solana side. \textit{Wormhole} relies on a set of guardian nodes to attest cross-chain messages, generating a signed VAA (Verifiable Action Approval). On Solana, a \textit{Wormhole} program should verify guardian signatures before minting new tokens. However, a critical mistake was introduced in an update: the Solana verification function was using an outdated system call for signature checks, effectively bypassing the actual verification. The attacker discovered that they could simply fabricate a data account that pretended to be a valid signature set, since the contract wasn’t actually validating the signatures against the guardian keys. By exploiting this signature verification bypass, the attacker called \textit{Wormhole}’s complete\_wrapped routine on Solana to mint 120k wrapped ETH for themselves with no real backing. The exploit structure is depicted in Figure~\ref{fig:wormhole-attack-overview}.

\begin{figure}[h]
    \centering
\includegraphics[width=0.9\linewidth]{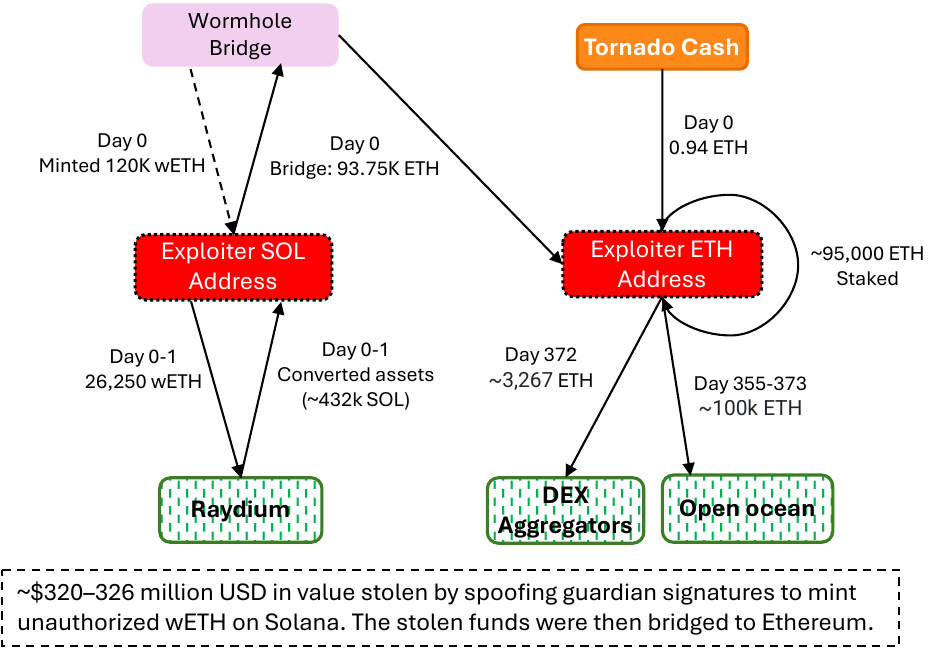}
    \caption{
    Attack flow of the \textit{Wormhole} bridge exploit. 
    }
    \label{fig:wormhole-attack-overview}
\end{figure}

\textit{Qubit Finance}’s bridge between Ethereum and BSC was hit for \$80 million in an attack that abused improper input validation. \textit{Qubit}’s contract allowed users to deposit an asset on Ethereum and mint a pegged version on BSC. The hack targeted the BSC side’s deposit function: due to a logic error in a single overlooked condition, the bridge’s smart contract did not verify that a call to transfer ETH actually happened when a deposit message was processed. In simpler terms, an attacker invoked the bridge deposit with zero ETH but bypassed the intended failure check, tricking the contract into believing collateral was provided. This exploit allowed the attacker to mint limitless xETH (the bridged ETH on BSC) without any real ETH locked on Ethereum. The exploit structure is depicted in Figure~\ref{fig:qubit-attack-overview}.

\begin{figure}[h]
    \centering
\includegraphics[width=0.9\linewidth]{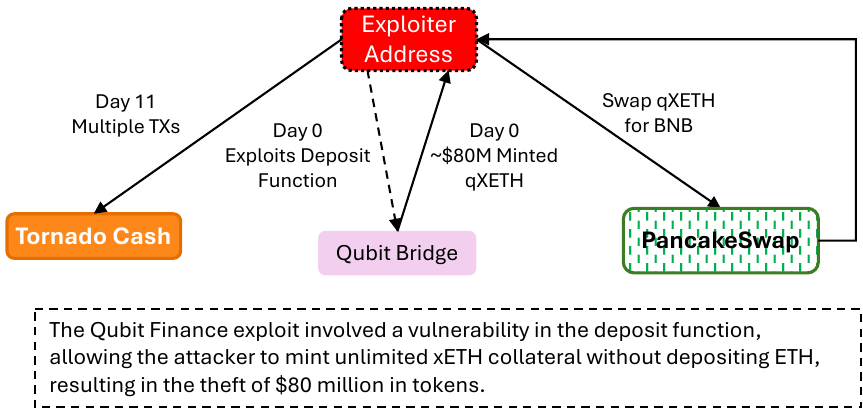}
    \caption{
    Attack flow of the \textit{Qubit} bridge exploit.
    }
    \label{fig:qubit-attack-overview}
\end{figure}

The \textit{Nomad} bridge (connecting Ethereum with Moonbeam and other chains) suffered a \$190 million loss in one of the most chaotic exploits to date. \textit{Nomad} was an optimistic messaging bridge, meaning it had a system of bonded updaters and a fraud challenge window. Unfortunately, after a routine smart contract upgrade, a critical initialization bug was introduced: the \textit{Nomad} Replica contract on Moonbeam was set with a trusted root of 0x00…00 (zero) by mistake, which immediately enabled every message to be accepted as valid.  Attackers quickly realized this and began submitting bogus withdrawal messages to drain tokens. The incident became a free-for-all: once the method was public, dozens of copycats (some white-hat, some black-hat) jumped in to also withdraw funds by replaying the attacker’s transaction call data. The drain pattern across wallets is illustrated in Figure~\ref{fig:nomad-attack-overview}.

\begin{figure}[h]
    \centering
\includegraphics[width=0.9\linewidth]{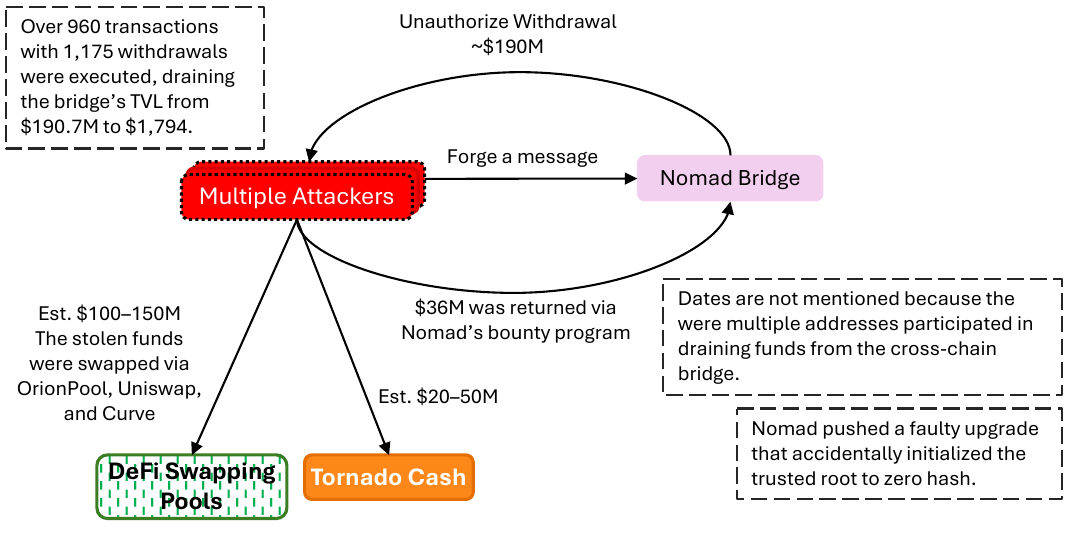}
    \caption{
    Attack flow of the \textit{Nomad} bridge exploit. 
    }
    \label{fig:nomad-attack-overview}
\end{figure}

Harmony’s \textit{Horizon} bridge (connecting Harmony’s chain with Ethereum) was hacked for about \$100 million in assets, through a straightforward attack on its 2-of-5 multisig validators. Similar to \textit{Ronin}, the attackers somehow obtained the private keys for at least two of the five signer addresses that controlled the bridge. Figure~\ref{fig:horizon-attack-overview} shows the flow of compromised validator-controlled funds.

\begin{figure}[h]
    \centering
\includegraphics[width=0.9\linewidth]{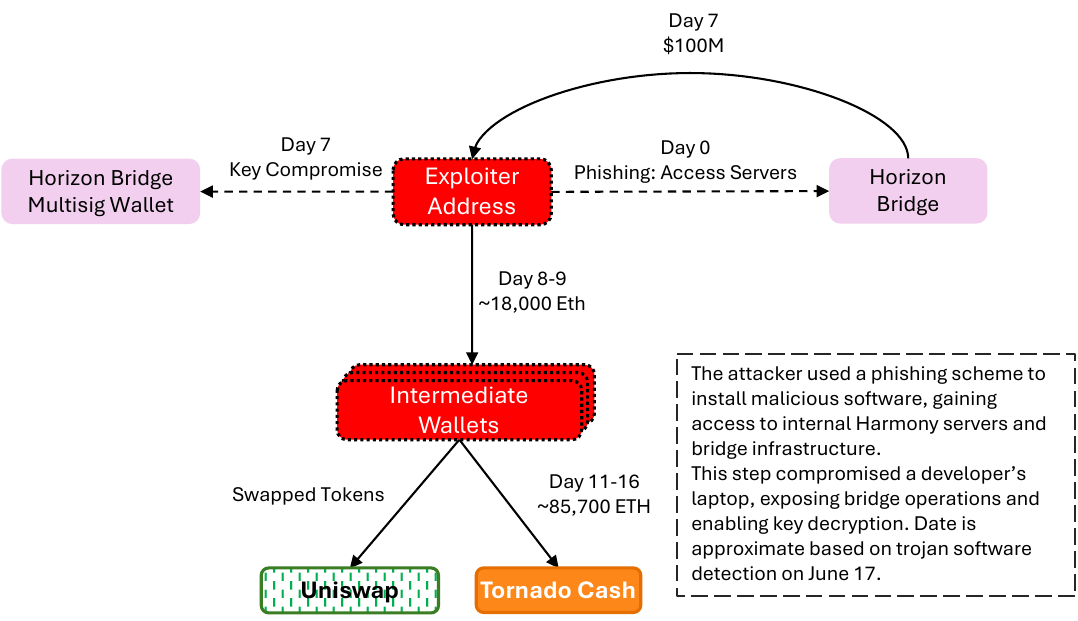}
    \caption{
    Attack flow of the \textit{Horizon} bridge exploit. 
    }
    \label{fig:horizon-attack-overview}
\end{figure}

An attacker exploited the Binance Smart Chain’s main bridge (between BSC and its Beacon Chain) and managed to mint 2 million BNB ($\approx$ \$586 million). The root cause was a bug in the light-client based proof verification on the BSC side. The bridge used an internal light client (IAVL tree-based) to verify bridge messages. The attacker created a fake proof for a past block that the BSC light client would accept as valid, bypassing the normal transaction inclusion checks. By forging this cryptographic proof, the attacker was able to inject a malicious block update that instructed the bridge to mint new BNB to their address. It was a complex exploit at the cryptography/verification layer, essentially an attack on the consensus verification between the two Binance chains. In response, Binance halted the entire BSC chain for 8 hours to prevent the attacker from moving more funds. A large portion of the illicit BNB never left Binance-controlled wallets, limiting the realized theft to around \$100M. This incident revealed that even a seemingly trustless bridge (a light-client-based one) can have implementation flaws in its verification logic. The exploit structure is depicted in Figure~\ref{fig:binance-attack-overview}.

\begin{figure}[h]
    \centering
\includegraphics[width=0.9\linewidth]{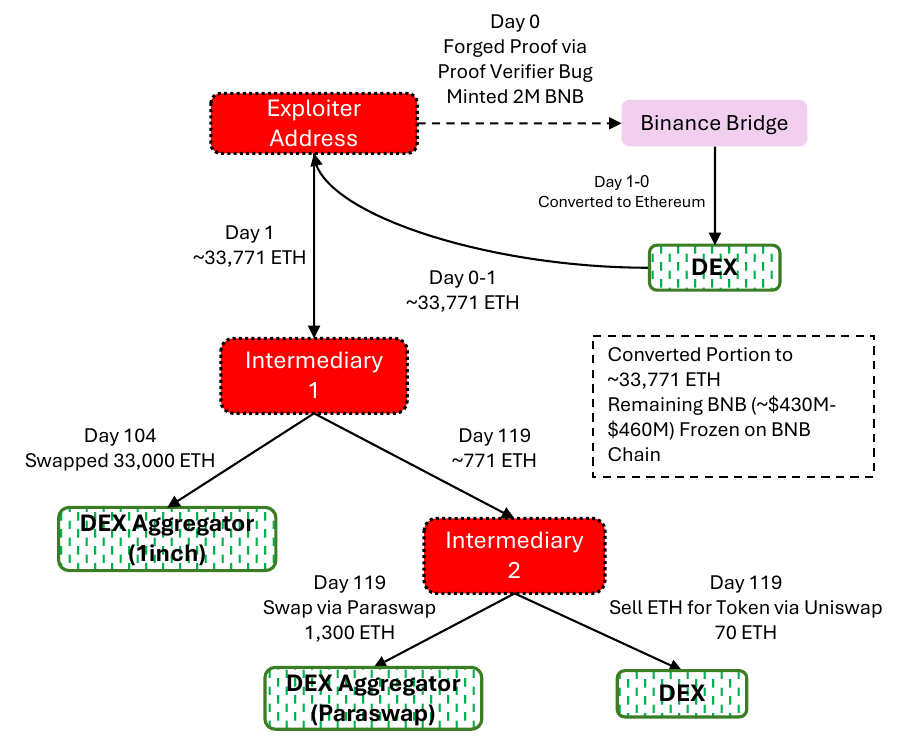}
    \caption{
    Attack flow of the Binance exploit. 
    }
    \label{fig:binance-attack-overview}
\end{figure}

In mid-2023, the \textit{Multichain} bridge (formerly Anyswap), which linked over a dozen EVM chains, experienced a sudden breach resulting in approximately \$126 million in assets withdrawn. Subsequent investigation suggested a compromise of the project’s private keys or server infrastructure. It came to light that all of \textit{Multichain}’s critical MPC key shares were under the control of its CEO, who had been detained by authorities in China. The attackers (or insiders) managed to use these keys to authorize massive transfers from the bridge’s liquidity pools on multiple chains. Essentially, this was an insider compromise. The fallout also disrupted many linked chains’ assets and led to the project’s collapse. The fund outflow pattern across chains is shown in Figure~\ref{fig:multichain-attack-overview}.

\begin{figure}[h]
    \centering
\includegraphics[width=0.9\linewidth]{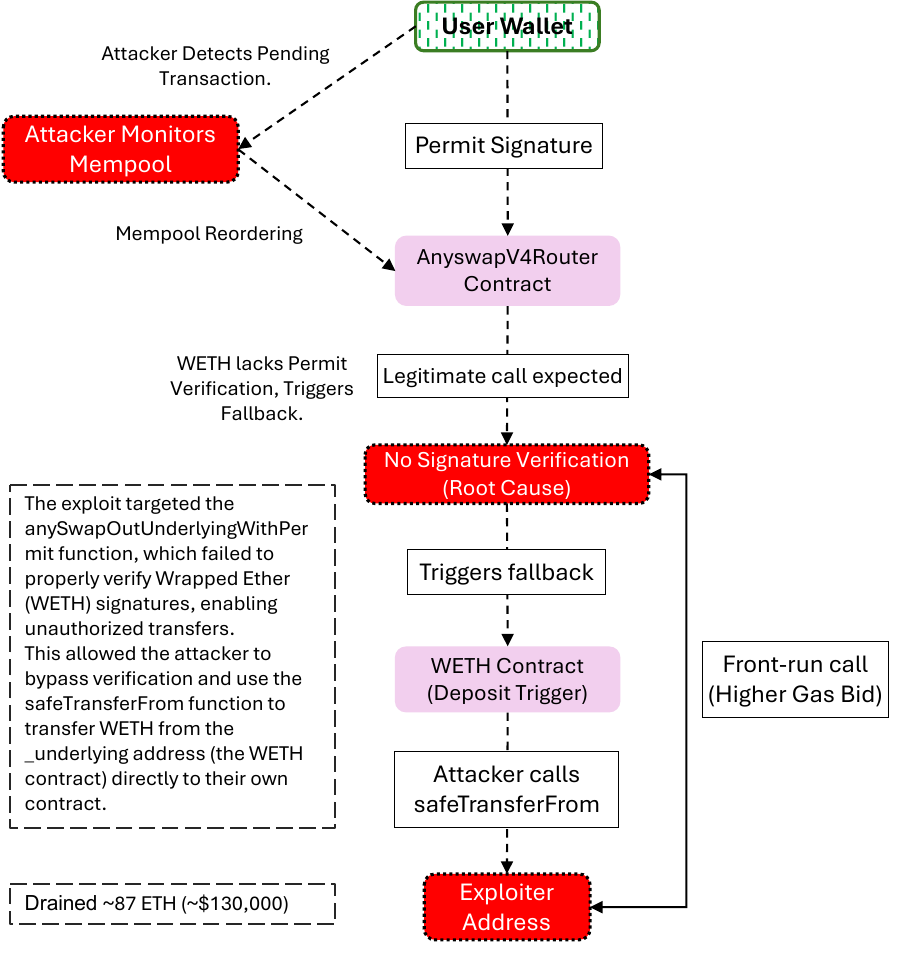}
    \caption{
    Attack flow of the \textit{Multichain} exploit. 
    }
    \label{fig:multichain-attack-overview}
\end{figure}

Following Ethereum’s transition from proof-of-work to proof-of-stake, the OmniBridge was exploited in September 2022 using a replay attack. Initially, the attacker sent 200 WETH using the OmniBridge on the Ethereum PoS chain. Since  transaction format was identical did not have chain ID verification, the attacker replayed the same transaction on the EthereumPoW fork. OmniBridge failed to distinguish between the two chains, therefore it processed the transaction and then released 200 ETHW to the attacker again. This exploit highlights the need for implementing strict chain ID validation during hard forks or network upgrades due to the possibility  of having identical transaction histories across chains.

\begin{figure}[h]
    \centering
\includegraphics[width=0.9\linewidth]{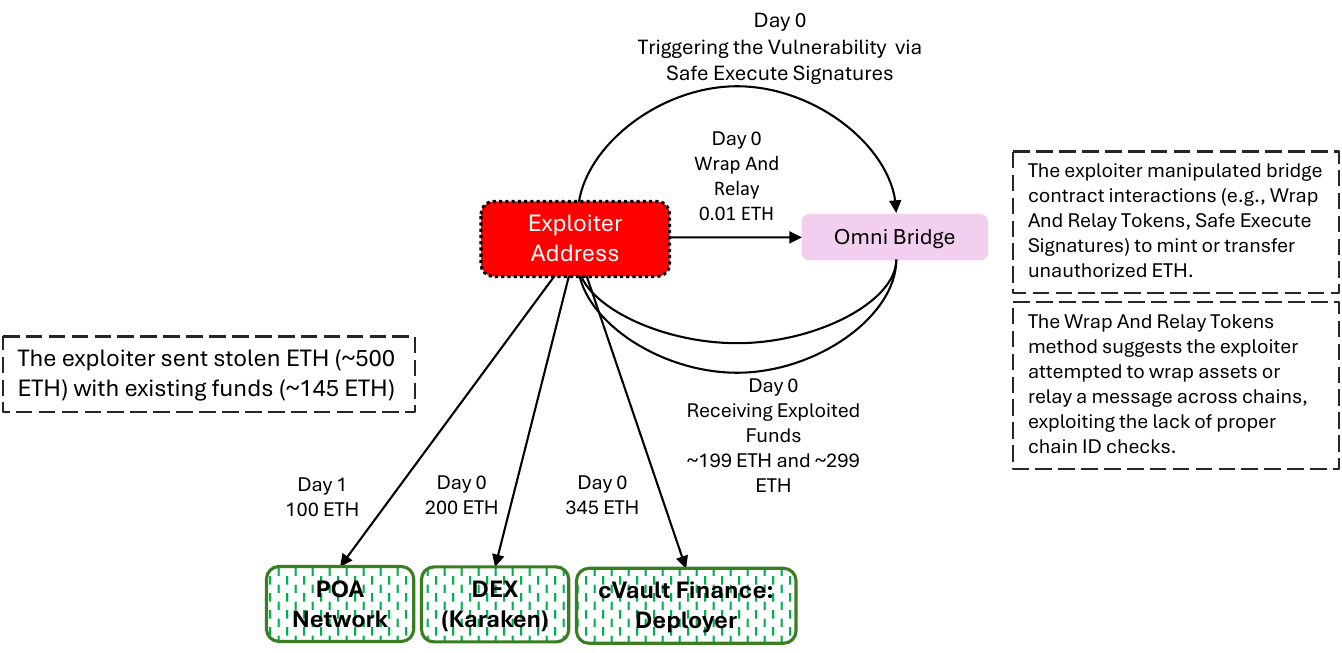}
    \caption{
    Attack flow of the \textit{Omni} bridge exploit. 
    }
    \label{fig:omnibridge-attack-overview}
\end{figure}

Smaller but instructive bridge attacks continued through 2024. For example, \textit{Orbit Chain} (Jan 2024) lost ~\$10M when 7 of its 10 bridge validators were compromised, another case of federated signer failure. An exploit in \textit{ALEX Bridge} (May 2024) (connecting Stacks to other chains) also occurred, reportedly due to a logic bug in its code. These ongoing incidents show that despite industry awareness, bridge security remains challenging.

\begin{table*}[ht]
    \centering
    \caption{Blockhain Bridge Exploits}
    \label{tab:bridge_exploits}
    \renewcommand{\arraystretch}{1.2}
    \resizebox{\textwidth}{!}{ 
    \begin{tabular}{l l l l l}
        \toprule
        \textbf{Exploit Address} & \textbf{Attack Name} & \textbf{Attack Date} & \textbf{Bridge Name} & \textbf{Bridge Address} \\
        \midrule
        \href{https://etherscan.io/address/0xc8a65fadf0e0ddaf421f28feab69bf6e2e589963}{0xc8a...963} & Poly Network Trusted State Root Exploit 1 & 2021-08-10 & Poly Network EthCrossChainManager & \href{https://etherscan.io/address/0x14413419452aaf089762a0c5e95ed2a13bbc488c}{0x144...88c} \\
        \href{https://etherscan.io/address/0x5dc3603c9d42ff184153a8a9094a73d461663214}{0x5dc...214} & Poly Network Trusted State Root Exploit 2 & 2021-08-10 & Poly Network EthCrossChainManager & \href{https://etherscan.io/address/0x14413419452aaf089762a0c5e95ed2a13bbc488c}{0x144...88c} \\
        \href{https://etherscan.io/address/0x1234567890abcdef1234567890abcdef12345678}{0x123...678} & Multichain Rush Attack 1 & 2023-02-15 & Multichain Router V4 & \href{https://etherscan.io/address/0x6b7A87899490EcE95443E979ca9485cbe7E71522}{0x6b7...522} \\
        \href{https://etherscan.io/address/0x9d5765ae1c95c21d4cc3b1d5bba71bad3b012b68}{0x9d5...b68} & Multichain Rush Attack 2 & 2023-02-15 & Multichain Router V4 & \href{https://etherscan.io/address/0x6b7A87899490EcE95443E979ca9485cbe7E71522}{0x6b7...522} \\
        \href{https://etherscan.io/address/0xefeef8e968a0db92781ac7b3b7c821909ef10c88}{0xefe...c88} & Multichain Rush Attack 3 & 2023-02-15 & Multichain Router V4 & \href{https://etherscan.io/address/0x6b7A87899490EcE95443E979ca9485cbe7E71522}{0x6b7...522} \\
        \href{https://etherscan.io/address/0x418ed2554c010a0c63024d1da3a93b4dc26e5bb7}{0x418...bb7} & Multichain Rush Attack 4 & 2023-02-15 & Multichain Router V4 & \href{https://etherscan.io/address/0x6b7A87899490EcE95443E979ca9485cbe7E71522}{0x6b7...522} \\
        \href{https://etherscan.io/address/0x622e5f32e9ed5318d3a05ee2932fd3e118347ba0}{0x622...ba0} & Multichain Rush Attack 5 & 2023-02-15 & Multichain Router V4 & \href{https://etherscan.io/address/0x6b7A87899490EcE95443E979ca9485cbe7E71522}{0x6b7...522} \\
        \href{https://etherscan.io/address/0x48bead89e696ee93b04913cb0006f35adb844537}{0x48b...537} & Multichain Rush Attack 6 & 2023-02-15 & Multichain Router V4 & \href{https://etherscan.io/address/0x6b7A87899490EcE95443E979ca9485cbe7E71522}{0x6b7...522} \\
        \href{https://etherscan.io/address/0x027f1571aca57354223276722dc7b572a5b05cd8}{0x027...cd8} & Multichain Rush Attack 7 & 2023-02-15 & Multichain Router V4 & \href{https://etherscan.io/address/0x6b7A87899490EcE95443E979ca9485cbe7E71522}{0x6b7...522} \\
        \href{https://etherscan.io/address/0x489a8756c18c0b8b24ec2a2b9ff3d4d447f79bec}{0x489...bec} & Binance Bridge Proof Verifier Bug & 2022-10-06 & Binance Bridge Ethereum Contract & \href{https://etherscan.io/address/0x69F4201EE81d155971AcC695AE5963eE8798366D}{0x69F...66D} \\
        \href{https://etherscan.io/address/0x82faed2da812d2e5cced3c12b3baeb1a522dc677}{0x82f...677} & Omni Bridge ChainID Vulnerability Exploit & 2022-09-18 & OmniBridge Multi-Token Mediator & \href{https://etherscan.io/address/0x88ad09518695c6c3712AC10a214bE5109a655671}{0x88a...671} \\
        \href{https://etherscan.io/address/0xb5c55f76f90cc528b2609109ca14d8d84593590e}{0xb5c...90e} & Nomad Trusted State Root Exploit 1 & 2022-08-01 & Nomad BridgeRouter & \href{https://etherscan.io/address/0x88a69b4e698a4b090df6cf5bd7b2d47325ad30a3}{0x88a...0a3} \\
        \href{https://etherscan.io/address/0x56D8B635A7C88Fd1104D23d632AF40c1C3Aac4e3}{0x56D...4e3} & Nomad Trusted State Root Exploit 2 & 2022-08-01 & Nomad BridgeRouter & \href{https://etherscan.io/address/0x88a69b4e698a4b090df6cf5bd7b2d47325ad30a3}{0x88a...0a3} \\
        \href{https://etherscan.io/address/0xBF293D5138a2a1BA407B43672643434C43827179}{0xBF2...179} & Nomad Trusted State Root Exploit 3 & 2022-08-01 & Nomad BridgeRouter & \href{https://etherscan.io/address/0x88a69b4e698a4b090df6cf5bd7b2d47325ad30a3}{0x88a...0a3} \\
        \href{https://etherscan.io/address/0x0d043128146654c7683fbf30ac98d7b2285ded00}{0x0d0...d00} & Horizon Bridge Private Key Compromised & 2022-06-24 & Horizon Bridge & \href{https://etherscan.io/address/0x2dccdb493827e15a5dc8f8b72147e6c4a5620857}{0x2dc...857} \\
        \href{https://etherscan.io/address/0x098b716b8aaf21512996dc57eb0615e2383e2f96}{0x098...f96} & Ronin Private Key Compromised (Social Engineering) & 2022-03-23 & Ronin Bridge Vault (V1) & \href{https://etherscan.io/address/0x1A2a1C938CE3eC39b6D47113c7955BAa9DD454F2}{0x1A2...4F2} \\
        \href{https://etherscan.io/address/0x629e7da20197a5429d30da36e77d06cdf796b71a}{0x629...71a} & Wormhole Account Spoofing & 2022-02-02 & Wormhole Portal Token Bridge & \href{https://etherscan.io/address/0x3ee18B2214AFF97000D974cf647E7C347E8fa585}{0x3ee...585} \\
        \href{https://etherscan.io/address/0xd01ae1a708614948b2b5e0b7ab5be6afa01325c7}{0xd01...5c7} & Qubit Finance Deposit Function Exploit & 2022-01-28 & Qubit QBridge (Ethereum) & \href{https://etherscan.io/address/0xD88E328C305F541E2dE6d3C85Ed081653Cd8A726}{0xD88...726} \\
        \href{https://etherscan.io/address/0x8c1944FAC705ef172f21f905b5523Ae260F76d62}{0x8c1...d62} & Thorchain Private Key Compromised (Phishing Attack) & 2021-07-24 & THORChain Bifrost ETH Router & \href{https://etherscan.io/address/0xc145990e84155416144c532e31f89b840ca8c2ce}{0xc14...2ce} \\
        \bottomrule
    \end{tabular}
    }
\end{table*}


\section*{Static Code Table Column Descriptions}
\label{sec:static-code-columns}

Tables~\ref{tab:accessControl}, \ref{tab:codeAndFunctionMetrics}, and \ref{tab:callsAndErrors} summarize features of bridge smart contracts extracted via static code analysis. Below, we define each column and its relevance.

\begin{itemize}
  \item \textbf{Local vars}: Number of local (function-scoped) variables, indicating complexity of internal logic.
  \item \textbf{Inheritances}: Number of Solidity contracts inherited, reflecting code reuse or modularity.
  \item \textbf{Modifier Count}: Number of \texttt{modifier} constructs used to restrict access or enforce invariants.
  \item \textbf{RoleBased}: Indicates whether role-based access control mechanisms (e.g., \texttt{AccessControl}) are implemented.
  \item \textbf{Standard Libs}: Count of imported standard libraries such as \texttt{SafeMath} or \texttt{Ownable}.
  \item \textbf{LOC (Lines of Code)}: Number of non-comment, non-whitespace lines in the source code.
  \item \textbf{Total Lines}: Full line count including comments and whitespace.
  \item \textbf{Public / External / Internal / Private Funcs}: Number of functions by visibility. Public and external functions are externally callable and represent direct attack surfaces.
  \item \textbf{Global vars Declared}: Number of state variables declared at the contract level.
  \item \textbf{Ext Funcs}: Number of externally callable functions.
  \item \textbf{Low-level}: Number of low-level calls (\texttt{call}, \texttt{delegatecall}, \texttt{staticcall}).
  \item \textbf{Untrusted}: Instances where the target of a low-level call is not a hardcoded or trusted address.
  \item \textbf{Reentry Guard}: Presence of reentrancy protection (e.g., via the \texttt{nonReentrant} modifier).
  \item \textbf{Require / Assert}: Total number of \texttt{require} and \texttt{assert} statements used for validation and safety checks.
  \item \textbf{Custom Errors}: Count of Solidity custom error definitions used for gas-efficient error handling.
  \item \textbf{Checks/Fn}: Average number of \texttt{require} or \texttt{assert} checks per function, used as a proxy for defensive programming density.
\end{itemize}

\end{document}